\documentclass[10pt,a4paper]{article}
\usepackage{amsmath,latexsym,amsfonts,amssymb}
\usepackage{wrapfig}
\usepackage{amscd,array,epic,eepic,calc,bbm,float}
\usepackage{ifthen,psfrag,verbatim,epsfig,graphicx,enumerate}
\usepackage{makeidx}
\usepackage{amsthm}
\usepackage{mathabx,bbm,bm}
\usepackage[dvipsnames]{xcolor}
\usepackage{natbib}
\usepackage[titletoc,title]{appendix}
\usepackage{subcaption}
\usepackage{hyperref}

\def\bql{\begin{equation}\label}
\def\eql{\end{equation}\noindent}

\def\brl{\begin{eqnarray}\label}
\def\erl{\end{eqnarray}\noindent}

\def\bro{\begin{eqnarray*}}
\def\ero{\end{eqnarray*}\noindent}

\def\brr{\begin{array}}
\def\err{\end{array}\noindent}

\def\bth{\begin{theorem}}
\def\eth{\end{theorem}}

\def\bcr{\begin{corollary}}
\def\ecr{\end{corollary}}

\def\bpr{\begin{proposition}}
\def\epr{\end{proposition}}

\def\blm{\begin{lemma}}
\def\elm{\end{lemma}}

\def\bdf{\begin{definition}}
\def\edf{\end{definition}}

\def\bas{\begin{assumptions}}
\def\eas{\end{assumptions}}

\def\bex{\begin{example}\rm}
\def\eex{\end{example}}

\def\bxx{\begin{exercise}\rm}
\def\exx{\end{exercise}}

\def\brm{\begin{remark}\rm}
\def\erm{\end{remark}}

\def\bmx{\begin{matrix}}
\def\emx{\end{matrix}}

\def\bma{\begin{pmatrix}}
\def\ema{\end{pmatrix}}

\def\bcs{\begin{cases}}
\def\ecs{\end{cases}}

\def\btb{\begin{center}\begin{tabular}}
\def\etb{\end{tabular}\end{center}}

\def\bit{\begin{itemize}}
\def\eit{\end{itemize}}

\def\a{\alpha}
\def\b{\beta}

\def\g{\gamma}

\def\l{\lambda}

\def\r{\rho}

\def\t{\tau}

\def\NC{{\cal N}}

\def\gb{{\bm g}}

\def\zerob{{\bf0}}

\def\thb{\boldsymbol{\theta}}
\def\lambdab{\boldsymbol{\lambda}}

\DeclareMathOperator*{\argmax}{arg\,max}

\def\ebb{{\mathbb E}}

\def\nf{\infty}

\def\tod{\buildrel\rm d\over\rightarrow}

\def\thesection{\arabic{section}}
\def\thesubsection{\arabic{section}.\arabic{subsection}}
\def\theequation{\arabic{section}.\arabic{equation}}
\def\thetheorem{\arabic{section}.\arabic{theorem}}

\def\theequation{\arabic{section}.\arabic{equation}}
\def\thetheorem{\arabic{section}.\arabic{theorem}}

\usepackage{framed}
\usepackage{comment}
\definecolor{shadecolor}{gray}{0.9}
\specialcomment{extra}{\begin{shaded}}{\end{shaded}}

\newtheorem{theorem}{Theorem}[section]
\newtheorem{lemma}[theorem]{Lemma}

\theoremstyle{definition}

\theoremstyle{definition}
\newtheorem{rmk}{Remark}
\theoremstyle{definition}
\newtheorem{example}{Example}

\newcommand{\bbE}{\mathbb{E}}
\newcommand{\mN}{\mathcal{N}}
\newcommand{\bzero}{\boldsymbol{0}}
\newcommand{\norm}[1]{\lVert#1\rVert}
\makeindex

\begin{document}
\date{}
\title{Bias-corrected empirical likelihood-based inference for the tail index under heavy-tailed models}
\author{Haodi Liang and Natalia Nolde}
\maketitle
% \eject
%\pagenumbering{roman}
\setcounter{page}{1}

%\begin{abstract}
%
%\end{abstract}
%\noindent{\bf Key words}: 

\setcounter{equation}{0}

\begin{abstract}
The tail index parameter of heavy-tailed probability models plays a key role in characterizing the tail decay of the underlying distribution function and is often involved in extrapolation procedures for various extreme value analysis questions. In this paper we revisit the question of tail index estimation and combine the ideas of bias-correction and empirical likelihood estimation to propose an estimator that offers an attractive alternative to some of the existing estimators. We develop an asymptotic theory for the proposed estimator and conduct simulation studies to demonstrate its performance in finite sample situations. The method is also applied to a data example for illustration.
\end{abstract}

\section{Introduction}

Heavy-tailed distributions are frequently used to model data from a variety of fields including finance, insurance and data networks; \cite{Resnick2007}. The tail behaviour of such probability models is characterized by the tail index, a parameter that captures the rate of tail decay of the underlying distribution function. Estimates of the tail index are often of interest in themselves signifying how heavy-tailed a model for a given dataset is. However, an estimate of the tail index is also often required for estimation of risk functionals such as 
high quantiles \citep{Weissman1978}, marginal expected shortfall \citep{Cai2014}, extreme conditional excess probabilities \citep{Abdous2005a}, extreme conditional quantiles \citep{NoldeZhang2020}, stress scenarios \citep{ZhouNolde2025}. %Extreme value approaches proposed in the above papers usually rely on extrapolation to go from an estimate that uses not too extreme data values to estimates of a risk functional associated with regions that contain few or no observations. The tail index frequently enters into the form of the extrapolation factor, thus driving both variability and bias of the resulting risk functional estimates. 
In this paper, we revisit the question of inference for the tail index in the setting of independent and identically distributed (i.i.d.) observations under the semi-parametric assumption of regular variation on the tail of the underlying distribution function.

%Regular variation is a common semi-parametric assumption used to represent heavy-tailed distributions. 
%In this paper, we revisit the question of inference for the tail index in the setting of independent and identically distributed (i.i.d.) observations.
%Estimation of the tail index has received a considerable attention in the literature. 
%There exist several approaches for estimation of the tail index, some of which we review below. 

Let $X_1,\ldots,X_n$ be i.i.d. random variables with a distribution function $F$ satisfying  the regular variation condition: $1-F(x) = x^{-1/\gamma}L_F(x)$ for $x>0$ and some $\g>0$, where $L_F(x)$ is a slowly varying function, i.e., for $x>0$, $L_F(tx)/L_F(t)\to1$ as $t\to\nf$. We write $1-F\in{\rm RV}_{-1/\g}$. Here $\gamma$ is the unknown tail index to be estimated. Equivalently, the quantile function admits the following representation:
\begin{equation} 
\label{reg.var}
F^{-1}\left(1-\frac{1}{x}\right) = x^{\gamma}L(x),\qquad x>0,
\end{equation}
where $L(\cdot)$ is again slowly varying.

Denote the order statistics based on $n$ observations by $X_{n,1} \leq \ldots \leq X_{n,n}$. Let $k = k_n$ be an intermediate sequence such that $k/n\to 0$, $k\to\infty$ as $n\to\infty$. A classic estimator of $\gamma$ is the Hill estimator \citep{Hill1975} given by
\begin{equation}\label{qhill}
\widehat\gamma_{H} = \frac{1}{k}\sum_{i=1}^{k} \log X_{n,n-k+i} - \log X_{n,n-k}.
\end{equation}
Under regularity conditions, $\widehat\gamma_H$ is consistent and asymptotically normal: $\sqrt{k}(\widehat\gamma_H - \gamma) \tod \NC(0,\gamma^2)$ for $n\to\nf$. This leads to an approximate $100(1-\alpha)\%$ confidence interval of the form:
%Under some general conditions, $\widehat\gamma_H$ is consistent. Under some additional second-order condition, $\widehat\gamma_H$ is asymptotically normal with
%$\sqrt{k}(\widehat\gamma_H - \gamma_0) \to N(0,\gamma_0^2),$ where $\gamma_0$ denotes the true value of the tail index.
%Therefore, a $100(1-\alpha)\%$ confidence interval for $\gamma_0$ based on the normal approximation can be constructed by
\begin{equation}
\label{ci.hill}
\mathcal{I}_{H}(\a) = [\widehat{\gamma}_H - z_{1-\alpha/2}\frac{\widehat{\gamma}_H}{\sqrt{k}},\widehat{\gamma}_H + z_{1-\alpha/2}\frac{\widehat{\gamma}_H}{\sqrt{k}}],
\end{equation}
where $z_{1-\alpha/2}$ is the $(1-\alpha/2)$-quantile of the standard normal distribution. The performance of the Hill estimator and confidence interval in \eqref{ci.hill} is often sensitive to the choice of sample fraction~$k$. While various methods exist to guide selection of~$k$ (see, e.g., \cite{DreesKaufmann1998,Danielsson_etal2001,Beirlant2002}), their effectiveness depends on the properties of the underlying distribution. The behaviour of the slowly varying function $L(\cdot)$ in~\eqref{reg.var} can lead to substantial bias of $\widehat\g_H$, which also affects coverage properties of $\mathcal{I}_{H}$, especially when sample size $n$ is not sufficiently large.  
 
Empirical likelihood (EL), as an alternative method for inference, has low model-misspecification risk and allows construction of confidence regions that are data-driven and range-respecting.
 %Alternatively, \cite{Lu2002} proposed an EL based method for estimating the tail index.
%The EL method allows to construct data-driven and range respecting confidence intervals.
%We next review the empirical likelihood method proposed by \cite{Lu2002} proposed for estimating $\gamma$.
\cite{Lu2002} propose the following procedure for deriving an EL-based confidence interval for tail index $\g$.  
Let $Y_j = j(\log X_{n,n-j+1} - \log X_{n,n-j})$ for $j = 1,\ldots,k$. Then the Hill estimator can alternatively be written as
\[\widehat{\gamma}_H = \frac{1}{k} \sum_{j=1}^{k} Y_j.
\]
It follows from \cite{Weissman1978} that, for $n$ large, $Y_j$'s ($j=1,\ldots,k$) are approximately i.i.d. exponential random variables with mean $\gamma$.
Now regard $\widehat\gamma_H$ as the mean of i.i.d. random variables $Y_1,\ldots,Y_k$ without any parametric model assumptions.
%Let $y_1,\ldots, y_k$ be realizations of $Y_1,\ldots,Y_k$. 
The empirical log-likelihood ratio function of $\gamma$ is given by
\[W_n(\gamma) = \sup_{p_1,\ldots,p_k}\left\{\sum_{j=1}^{k} \log(kp_j): \sum_{j=1}^{k} p_j = 1,\, \sum_{j=1}^{k}p_j (Y_j-\gamma) = 0,\, p_1,\ldots,p_k > 0 \right\}.
\]
Under certain general conditions, \cite{Lu2002} show that
\[-2W_n(\gamma) \tod \chi^2_1,\qquad n\to\nf.
\]
Then a $100(1-\alpha)\%$ asymptotic EL-based confidence interval for $\gamma$ can be constructed as
\begin{equation}
\label{ci.el}
\mathcal{I}_{EL}(\a) = \{\gamma: -2W_n(\gamma) < \chi^2_{1,1-\alpha}\},
\end{equation}
where $\chi^2_{1,1-\alpha}$ is the $(1-\alpha)$-quantile of the $\chi^2_1$ distribution. 

Numerical experiments in \cite{Lu2002} indicate that confidence intervals constructed using the empirical likelihood tend to provide better coverage than the intervals based on the normal approximation in~\eqref{ci.hill}.
%The simulation studies in \cite{Lu2002} indicate that EL confidence interval has better coverage precision compared to the confidence interval based on normal approximation, especially when $k$ is small. 
Coverage properties of $\mathcal{I}_{EL}$ can be further improved using the adjusted empirical likelihood (AEL) as suggested by \cite{Li2019}.
%\cite{Li2019} used the adjusted empirical likelihood (AEL) to further improve the coverage probability of the EL confidence intervals.

However, both the EL and AEL confidence intervals will have low coverage probabilities when the Hill estimator displays a bias in finite sample settings. This suggests exploring the possibility of combining the EL-based method for tail index estimation and confidence interval construction with bias correction. 

One approach to correcting for bias in the estimation of the tail index is based on an exponential regression model as proposed by \cite{Beirlant1999}, in which the authors impose an additional second-order condition:
%\cite{Beirlant1999} propose an approach for constructing a bias-corrected estimator of the tail index under an additional second-order condition that specifies the bias term:
%\cite{Beirlant1999} proposed a bias-corrected estimator for the tail index, which has low bias across a wide range of $k$ values, with an additional second-order condition that specifies the bias term:

\noindent\textbf{Condition ($R_L$)}:
There exists a real constant $\rho < 0 $ and a rate function $h(\cdot)$ satisfying
$h(x) \to 0$ as $x\to \infty$, such that for all $t \ge 1,$ as $x\to\infty$,
\[\log \frac{L(t x)}{L(x)} \sim h(x) \frac{t^{\rho}-1}{\rho}.
\]
Under this condition, \cite{Beirlant1999} show that, for $n$ large, the log-spacings $Y_1,\ldots,Y_k$ are approximately independent and exponentially distributed:
\begin{equation}
\label{approximation}
Y_j \stackrel{\cdot}{\sim} Exp\Big(\Big[\gamma + b_{n,k} \Big(\frac{j}{k+1}\Big)^{-\rho}\Big]^{-1}\Big),\qquad j=1,\dots,k
\end{equation}
with
\[b_{n,k} = h\Big(\frac{n+1}{k+1}\Big), \qquad 2\le k\le n-1.
\]
They propose to estimate tail index $\gamma$ and nuisance parameters $b_{n,k}$ and $\rho$ jointly using the maximum likelihood approach. 

%The method is effective in reducing the bias of the Hill estimator, however, this comes at the expense of a larger variance.   

%{\color{red}[TO RE-WRITE]Although  the bias of the estimation of $\gamma$ is significantly reduced, when $\rho$ is close to 0, the variance of the MLE of $\gamma$ increases drastically \cite{Beirlant1999}. To this end, they optimize the likelihood function within a constrained range of $\rho$.   Moreover, the MLE of $\rho$ through joint estimation is generally inconsistent. This would result in inaccurate estimation of the asymptotic variance of the MLE of $\gamma$ and thus lead to inaccurate coverage probabilities of the corresponding confidence intervals.}

When $\rho$ can be estimated consistently, the asymptotic variance of the maximum likelihood estimator of $\g$ from the exponential regression model is $((1-\r)/\r)^2$ that of the Hill estimator \citep{Beirlant2002}. So bias reduction comes at the cost of increased variance. Additionally, the maximum likelihood estimator of $\rho$ through the joint estimation is generally inconsistent (see Remark~4 in \cite{Beirlant2002}), leading to inaccurate estimation of the asymptotic variance, which will impact the coverage properties of confidence intervals for $\g$. 

Instead of assuming an exponential model on $Y_j$'s in \eqref{approximation},
we consider a method which only relies on the associated mean approximation:
\[\ebb(Y_j) = \gamma + b_{n,k}\left(\frac{j}{k+1}\right)^{-\rho},\quad j=1,\ldots,k.
\]
 Furthermore, rather than dealing with three parameters ($\g$, $b_{n,k}$ and $\r$), we propose to replace the estimator of $\r$ with a canonical value, motivated by a similar strategy adopted in \cite{Feuerverger1999}. % {\color{red} Moreover, our method leads to a simple chi-squared  limiting distribution of the EL ratio statistic and does not require explicit estimation of the asymptotic variance}. 
Simulation studies show that this works reasonably well across a range of probability distributions with reduction in variance and improved coverage properties of confidence intervals compared to the bias-corrected MLE of \cite{Beirlant1999}. 

The remainder of the paper is organized as follows. The methodology is described in Section~\ref{s.method}. We establish an asymptotic normality result for the proposed estimator of the tail index and show that the EL ratio statistic has a simple chi-squared  limiting distribution. The latter result serves as a basis for constructing the bias-corrected EL confidence intervals.
In Section~\ref{s.sim}, we report results of simulation studies that assess finite sample performance of the proposed estimator and empirical coverage probabilities of the associated confidence intervals. Section~\ref{s.appl} presents an application of the method to a real life data example. Concluding remarks are provided in Section~\ref{s.concl}. Proofs are relegated to the Appendix.

\section{Methodology}\label{s.method}

In this section, we introduce a bias-corrected empirical likelihood estimator of the tail index. It is constructed using estimating functions that arise from the least squares objective function motivated by the approximation in~\eqref{approximation}. 

Based on~\eqref{approximation}, we have the following mean approximation:
\[\ebb(Y_j)\approx \gamma + b_{n,k}\left(\frac{j}{k+1}\right)^{-\rho},\qquad j=1,\ldots,k.
\]
Since the second-order parameter $\rho$ is difficult to estimate in general, an alternative is to replace its value by a canonical choice $\r_c<0$, as was suggested in  \cite{Feuerverger1999}. This leads to the least squares objective function of the form:

\[S(\thb) = \sum_{j=1}^{k} \Big (Y_j - \gamma - b_{n,k} \Big(\frac{j}{k+1}\Big)^{-\rho_c}\Big)^2,\qquad \thb = (\g,b_{n,k}).
\]
The partial derivatives of $S(\thb)$ can be used to specify the estimating functions in the definition of the empirical likelihood. Here, the partial derivatives of $S(\thb)$ are given by
\[\frac{\partial S(\thb)}{\partial\gamma } = -2\sum_{j=1}^{k} \left\{ Y_j - \gamma - b_{n,k}\Big(\frac{j}{k+1}\Big)^{-\rho_c}\right\},\quad\frac{\partial S(\thb)}{\partial b_{n,k}} = -2\sum_{j=1}^{k}\left\{Y_j - \gamma - b_{n,k}\Big(\frac{j}{k+1}\right)^{-\rho_c}\Big\}\Big(\frac{j}{k+1}\Big)^{-\rho_c}.
\]
Now let
\begin{align*}
g_{1j}(Y_j;\thb) &= Y_j - \gamma -  b_{n,k}\left(\frac{j}{k+1}\right)^{-\rho_c},\\
g_{2j}(Y_j;\thb) &= \left (Y_j - \gamma - b_{n,k}\left(\frac{j}{k+1}\right)^{-\rho_c}\right )\left(\frac{j}{k+1}\right)^{-\rho_c},
\end{align*}
and set
\[\gb_j(Y_j;\thb) = \begin{pmatrix}
g_{1j}(Y_j;\thb)\\
g_{2j}(Y_j; \thb)\\
\end{pmatrix},\qquad j=1,\ldots,k.
\]
The empirical log-likelihood function of $\thb$ using $g_{1j}$ and $g_{2j}$ ($j=1,\ldots,k$) above as the estimating functions is equal to
\begin{equation}
\label{optimization}
\ell_{E}(\thb) = \sup_{p_1,\ldots,p_k}\Big\{\sum_{j=1}^{k}\log p_j: p_1,\ldots,p_k\ge0,\;\sum_{j=1}^{k} p_j = 1,\; \sum_{j=1}^{k} p_j \gb_j(Y_j;\thb) = \bzero
\Big\}.
\end{equation}
We then define a bias-corrected maximum empirical likelihood estimator (MELE) of $\thb = (\g,b_{n,k})$ as the maximizer of the above empirical  log-likelihood:
\begin{equation}\label{qMELE}
\widehat\thb_{E} = (\widehat\gamma_E,\widehat b_E) = \argmax_{\thb} \ell_{E}(\thb).
\end{equation}

The standard empirical likelihood theory requires the estimating functions to be unbiased \citep{QinLawless1994}. Replacing $\rho$ by $\rho_c$ in the least square objective function $S(\thb)$ leads to biased estimating functions $g_{1j}$ and $g_{2j}$. Under certain conditions on $b_{n,k},$ the bias term is sufficiently small so that the asymptotic properties of $\widehat\thb_E$ and the chi-squared  limiting distribution of the empirical likelihood ratio statistic still hold, as will be shown later.

The unique solution to the optimization problem in~\eqref{qMELE} exists provided the convex hull of points $\gb_j(Y_j;\thb)$ for $j=1,\ldots,k$ contains $\zerob$ \citep{Owen1990}, and is given by
\[p_j = \frac{1}{k[1+\boldsymbol{\lambda}^{\top}(\thb)\gb_j(Y_j;\thb)]},
\]
where $\boldsymbol{\lambda}(\thb)$ is the Lagrange multiplier satisfying
\begin{equation}
\label{Lagrange}
\sum_{j=1}^k \frac{1}{1+\boldsymbol{\lambda}^{\top}(\thb)\gb_j(Y_j;\thb)} = \bzero.
\end{equation}
This leads to the following expression for the empirical log-likelihood function:
\begin{equation}
\label{el.function}
\ell_E(\thb) = - k\log k - \sum_{j=1}^k \log[1+\boldsymbol{\lambda}^{\top}(\thb)\gb_j(Y_j;\thb)].
\end{equation}
When the convex hull of points $\gb_j(Y_j;\thb)$ does not contain $\zerob$, by convention $\ell_E(\thb) = -\infty.$

In the next theorem we prove asymptotic normality of $\widehat\thb_E$ and in particular that of $\widehat\g_E$ under an order condition on $b_{n,k}$. 

\begin{theorem}
\label{asymptotic.normality}
Consider a random sample $X_1,\ldots,X_n$ with distribution function $F$ satisfying $1-F \in RV_{-1/\gamma}$ for some $\g>0$ and Condition $(R_L)$ with rate function~$h$. Let $b_{n,k}=h\left(\dfrac{n+1}{k+1}\right)$ for $2\le k\le n-1$ and assume that, for any $\delta > 0$, $k^{1/2 + \delta} b_{n,k} = O(1)$ as $k,n\to \infty$ with $k/n \to 0$. Then the maximum empirical likelihood estimator $\widehat\g_E$ of $\g$ defined in \eqref{qMELE} satisfies:
\[\sqrt{k}(\widehat{\gamma}_E - \g) \overset{d}{\to} \mathcal{N}\left(0,\left(\frac{1-\rho_c}{\rho_c}\right)^2\g^2\right),\quad n\to\nf,\quad k=k_n\to\nf,\quad k/n\to0.\qquad
\]
\end{theorem}
The proof is given in Appendix~\ref{sA1}.

\begin{rmk}
Note that by replacing $\r$ in the objective function $S(\thb)$ with a canonical value $\r_c<0$, rather than a consistent estimator of $\r$, we have to impose a slightly stronger condition on the rate of decay of $b_{n,k}$:  $k^{1/2 + \delta} b_{n,k} = O(1)$ compared to $\sqrt{k} b_{n,k} = O(1)$ in \cite{Beirlant2002} (Theorem~3.2). However, there are potential gains in the asymptotic variance when $\r_c<\r$.
\end{rmk}

Theorem \ref{asymptotic.normality} serves as a basis for establishing the limiting distribution of the EL ratio statistic, which can then be used to construct confidence intervals for the tail index parameter. Unlike the bias-corrected parametric method, the empirical likelihood does not require explicitly estimating the asymptotic variance, and
one of the advantages of the empirical likelihood approach is the data-driven shape of the resulting confidence intervals (and regions).

The empirical log-likelihood ratio function for tail index $\g$ is given by
\begin{equation}
R(\g) := -2 \left[\max_{b} \ell_E(\g,b) - \max_{\thb \in (0,\infty) \times \mathbb{R}}\ell_E(\thb)\right],
\end{equation}
where $\ell_E$ is the empirical log-likelihood function in~\eqref{optimization}. We state the main result in the following theorem giving the asymptotic behaviour of $R(\g_0)$, where $\g_0$ is the true value of the tail index.

\begin{theorem}\label{tLR}
Under the same assumptions as in Theorem~\ref{asymptotic.normality} with the true value of tail index parameter equal to $\g_0>0$,
\[R(\gamma_0)  \tod \chi_1^2,\qquad n\to\nf.
\]
\end{theorem}
See Apprendix~\ref{sA2} for the proof.

Theorem~\ref{tLR} can be used to construct a confidence interval for the tail index. We hence define a $100(1-\alpha)\%$ bias-corrected empirical likelihood (BCEL) confidence interval for $\gamma_0$ as
\begin{equation}
\label{bel.ci}
\mathcal{I}_{BCEL}(\a)=\{\gamma: R(\gamma) \leq \chi^2_{1,1-\alpha}\},
\end{equation}
where $\chi^2_{1,1-\alpha}$ is the $(1-\alpha)$th quantile of the $\chi^2_1$ distribution.

\begin{rmk}
In \cite{Lu2002}, the EL is defined based on $\bbE(Y_j) = \g.$ The EL function of $\g$ is well defined for $\g \in (\min_j Y_j, \max_j Y_j),$ so the corresponding EL confidence interval is range-respecting. After introducing a second order parameter $b_{n,k}$, when $n$ is small, the bias-corrected EL confidence intervals in~\eqref{bel.ci} may include negative values of $\g$. As $n$ becomes larger, $b_{n,k}$ tends to 0, and the range-respecting property of the EL confidence interval (region) will generally be retained.

%After introducing a second order parameter $b_{n,k}$, the EL ratio function $R(\g)$ may be well-defined for $\g < 0$ when $n$ is small, and
%the empirical likelihood confidence interval may lose its range-respecting property. As $n$ becomes larger, $b_{n,k}$ tends closer to 0, the range-respecting property of the EL confidence interval (region) is generally retained.
\end{rmk}

\section{Simulation study}\label{s.sim}

In this section, we conduct a simulation study to examine the finite sample performance of the proposed bias-corrected maximum empirical likelihood estimator across a range of probability distributions with different values of tail index $\gamma$ and second order parameter $\rho$. 

We consider the following two families of probability distributions to generate random samples:
\begin{itemize}
\item Student t distribution with $\nu$ degrees of freedom for $\nu\in\{2,3,4,5\}$. The tail index is $\g=1/\nu$ and the second order parameter is $\rho=-2/\nu$.
\item Burr distribution with shape parameters $\l$ and $\t$ and scale $\b=1$, whose survival function has the form:
$$1-F(x) = \left( \dfrac{1}{1+x^\t} \right)^\l,\qquad x>0,\qquad \l,\t>0. $$
We consider $\l\in\{1,4/3,2,3\}$ and set $\t=1/\l$. The tail index is $\g=1/(\t\l)$ and the second order parameter is $\rho=-1/\l$.
\end{itemize}

The choice of the parameter values above is made with the focus on cases where $-1\le\r<0$ for which the Hill estimator $\widehat\g_H=\widehat\g_H(k)$ is known to exhibit a bias over a range of $k$ values. This makes the Hill plot $\{(k,\widehat\g_H(k))\}$, a graphical tool for selecting the sample fraction $k$, difficult to interpret due to trends. \cite{Resnick2007} refers to these situations as Hill  horror plots. Figure~\ref{fig:HillHorrorPlots} illustrates a Hill plot for a random sample of size $n=500$ from a Student t distribution with $\nu=5$ degrees of freedom for which $\r=-2/5$. Without knowing the true value of tail index $\g$, it is difficult to choose a suitable range of $k$ values where the tail index estimates are stable.  

\begin{figure}[h]
    \centering
     \includegraphics[width=0.6\textwidth]{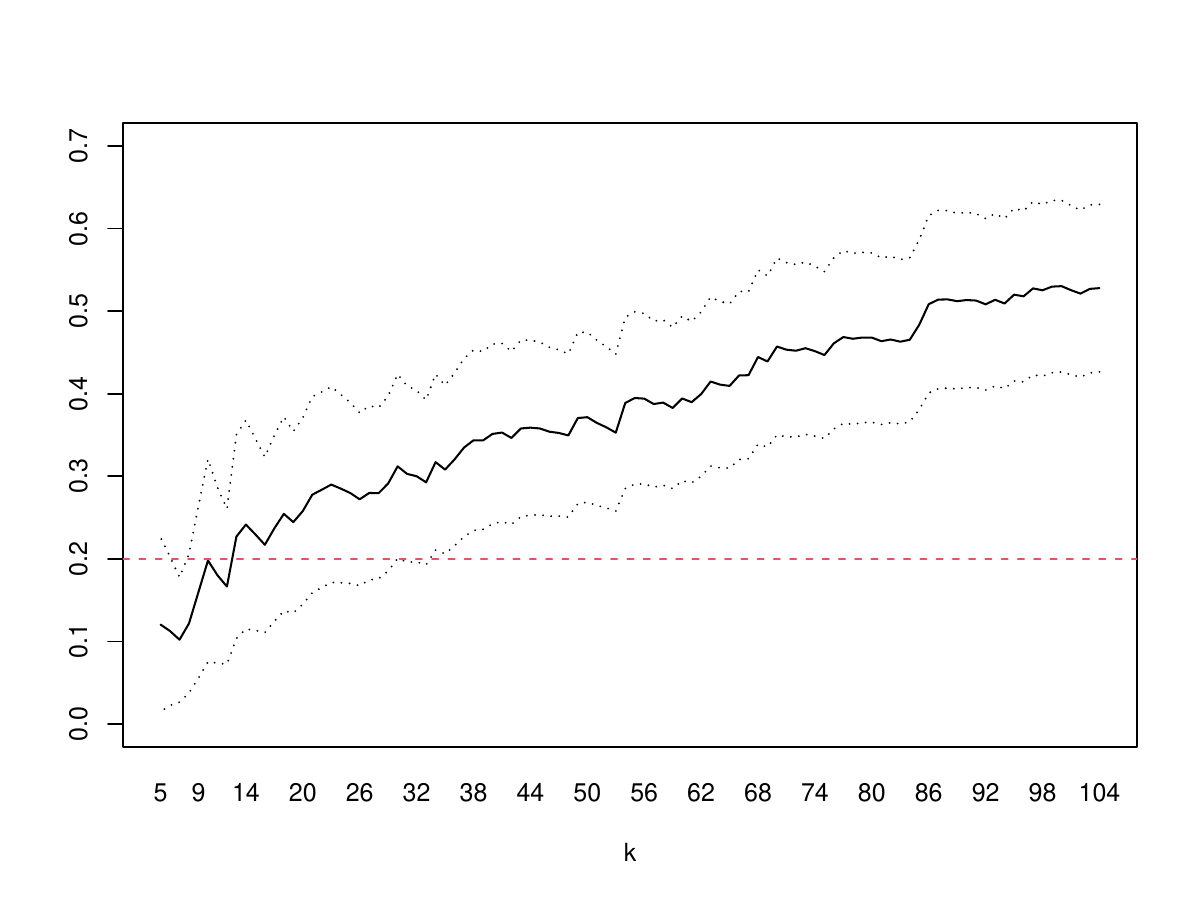}
    \caption{Hill plot for a random sample of size $n=500$ from a Student t distribution with degrees of freedom $\nu=5$. The tail index is $\gamma=0.2$ and second order parameter is $\rho=-0.4$. The dotted lines indicate approximate 95\% confidence intervals; the true value of $\g$ is indicated by the dashed red line.}
\label{fig:HillHorrorPlots}
\end{figure}

\subsection{Choice of the canonical value $\r_c$ in the bias-corrected MELE}
We begin by first examining the impact of the choice of the canonical value $\r_c$ in the proposed bias-corrected MELE of tail index~$\g$ in~\eqref{qMELE}. For each of the probability  distributions given above, we generate 1,000 random samples of size $n= 500$. We consider three values of $\r_c\in\{-2,-1,-0.5\}$. Figures~\ref{fig:rc.t} and~\ref{fig:rc.burr} summarize the behaviour of the empirical relative squared bias, variance and mean squared error (MSE) as a function of sample fraction~$k$ for the Student t and Burr distributions, respectively. As expected, based on the form of the asymptotic variance (see Theorem~\ref{asymptotic.normality}), the empirical variance is directly influenced to the value of $\r_c$ with smaller values of $\r_c$ leading to smaller variances. In our set-up, $\r_c=-2$ results in the smallest variance across the full range of considered $k$ values. For the Student t distributions (Figure~\ref{fig:rc.t}), in case~(a) when $\r=-1$, the bias is small and of similar magnitude for $\r_c=-2$ and $\r_c=-1$, but increases sharply for larger values of $k$ when $\r_c=-0.5$. The overall impact is that MSE is smallest when $\r_c=-2$, followed by $\r_c=-1$, driven by the ordering of variances. In case~(b) with $\r=-2/3$, $\r_c=-1$ results in the smallest bias and, for $k$ above 60, also leads to smaller MSE compared to $\r_c=-2$. In the remaining two cases with $\r=-1/2$ and $\r=-2/5$, we observe that while setting $\r_c=-0.5$ produces lowest bias for smaller values of $k$, due to variance differences, $\r_c=-1$ provides the most optimal compromise in terms of MSE. It is notable that when $\r_c=-1$ the squared bias remains quite stable over the entire range of considered $k$ values. 
For Burr distributions (Figure~\ref{fig:rc.burr}), when $\r\in\{-1,-3/4\}$, the bias is small for all three choices of $\r_c$ and ordering of MSE curves aligns with that of variances. However, for $\r=-1/2$ and especially $\r=-1/3$, we begin to see the impact of $\r_c$ on the bias which in turn affects the performance of the estimator from the MSE perspective. For the smallest value of $\r=-1/3$, setting $\r_c=-1$ offers the best gains in MSE. 

In Figures~\ref{fig:rcc.t} and~\ref{fig:rcc.burr}, we display the empirical coverage rates of 95\% confidence intervals based on~\eqref{bel.ci} across different values of $k$. The accuracy appears to be directly linked to the behaviour of bias. In particular, for the Student t distributions (Figure~\ref{fig:rcc.t}), setting $\r_c=-1$ leads to either comparable or better coverage rates relative to the other two choices of $\r_c$. We also see that when $\r_c=-1$ the coverage accuracy improves for larger values of $k$. In the case of Burr distributions (Figure~\ref{fig:rcc.burr}), coverage rates behave similarly across $\r\in\{-1,-3/4,-1/2\}$ cases with $\r_c=-2$ leading to slightly more accurate results for smaller values of $k$. In the right most panel for $\r=-1/3$ only $\r_c=-0.5$ setting maintains coverage rates close to the nominal level of 95\%.

While there appears to be no one single optimal choice of $\r_c$ across all simulation settings, the above results do support the choice of $\r_c=-1$ suggested by \cite{Feuerverger1999} as the one that provides a good balance between bias and variance. Furthermore, the relative performance tends to improve for moderate values of $k$. In terms of coverage precision, this choice works well for $\r\le -0.5$ settings, although the performance may deteriorate as $\r$ gets closer to zero as seen in the Burr examples. In practice, one can consider several $\r_c$ values to assess which choice leads to a more stable behaviour of the tail index estimates over a range of $k$ values. For example, as discussed above, when $\r$ is close to $-1$, there may be a benefit of taking $\r_c=-2$ due to both control of bias and variance.

%For each generated sample, we first  compute the bias-corrected MELEs with three different choices of canonical parameters $\rho_c= -1, -0.5, -2$. We then plot the mean of the 1000 bias-corrected MELEs with $\rho_c = -1,-0,5,-2$ in Figures \ref{Bias_t_rho} and \ref{Bias_Burr_rho}. We see that the choice of $\rho_c = -1$ leads to smallest bias for most selected distributions. The choice of $\rho_c = -0.5$ often leads to significant downward bias while the choice of $\rho_c = -2$ often leads to upward bias. This justifies the suggested choice of $\rho_c = -1$ in \cite{Feuerverger1999}. Figures \ref{coverage_t_rho} and \ref{coverage_Burr_rho} show the coverage rates of bias-corrected empirical likelihood confidence intervals with $\rho_c = -1, -0.5, -2$. It is seen that the choice of $\rho_c = -1$ leads to most accurate coverage rates.
%The above simulation results confirm that the canonical choice of $\rho_c = -1$ leads to the best numerical performance of the bias-corrected MELE among the three choices. 

\captionsetup[subfigure]{skip=2pt}

\begin{figure}[htbp]
\centering

\begin{subfigure}{0.83\textwidth}
    \centering
    \includegraphics[width=\linewidth]{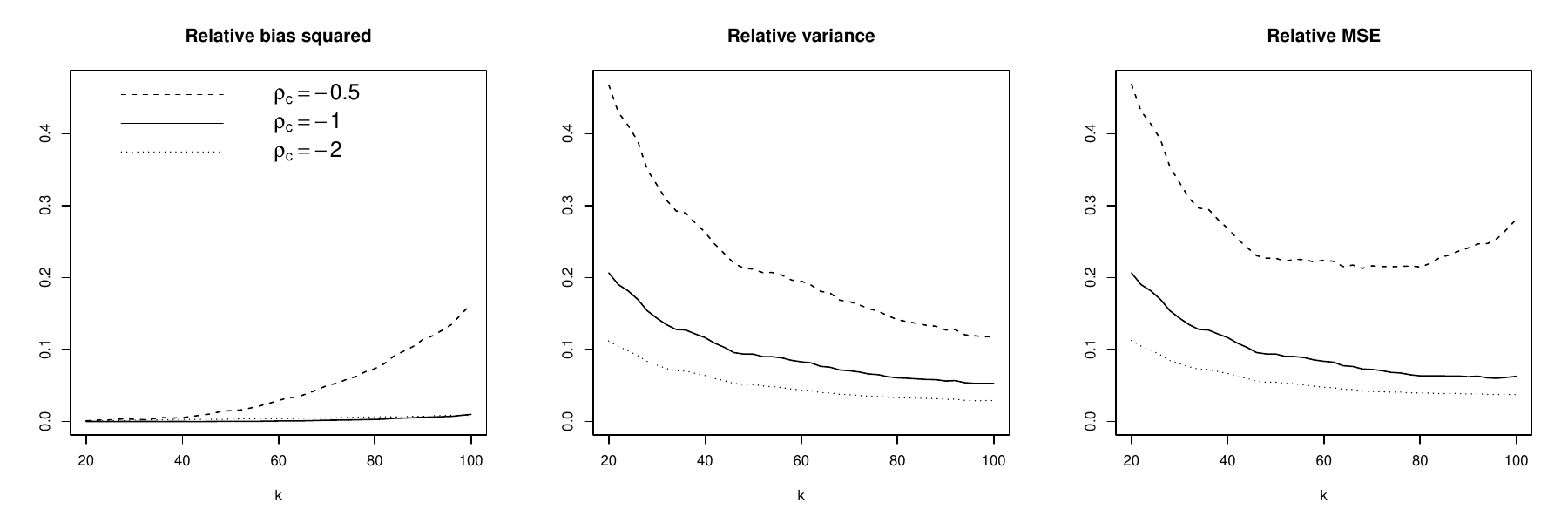}
    \caption{$\nu=2$, $\rho=-1$}
\end{subfigure}

%\vspace{0.5em}

\begin{subfigure}{0.83\textwidth}
    \centering
    \includegraphics[width=\linewidth]{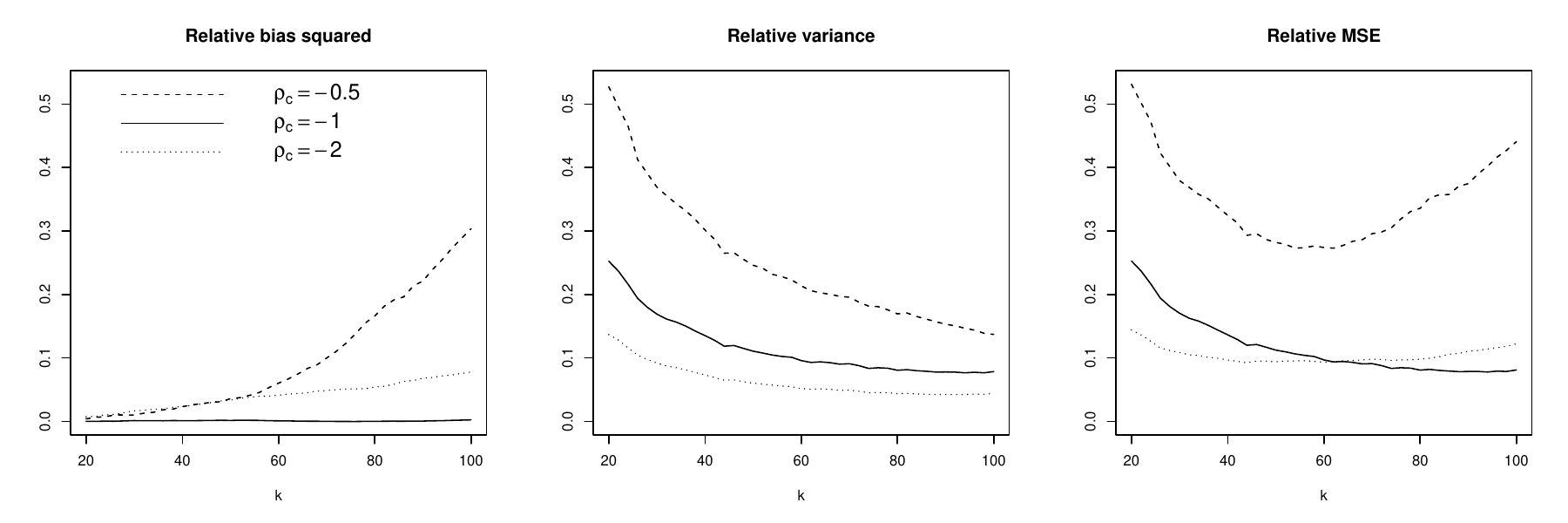}
    \caption{$\nu=3$, $\rho=-2/3$}
\end{subfigure}

%\vspace{0.5em}

\begin{subfigure}{0.83\textwidth}
    \centering
    \includegraphics[width=\linewidth]{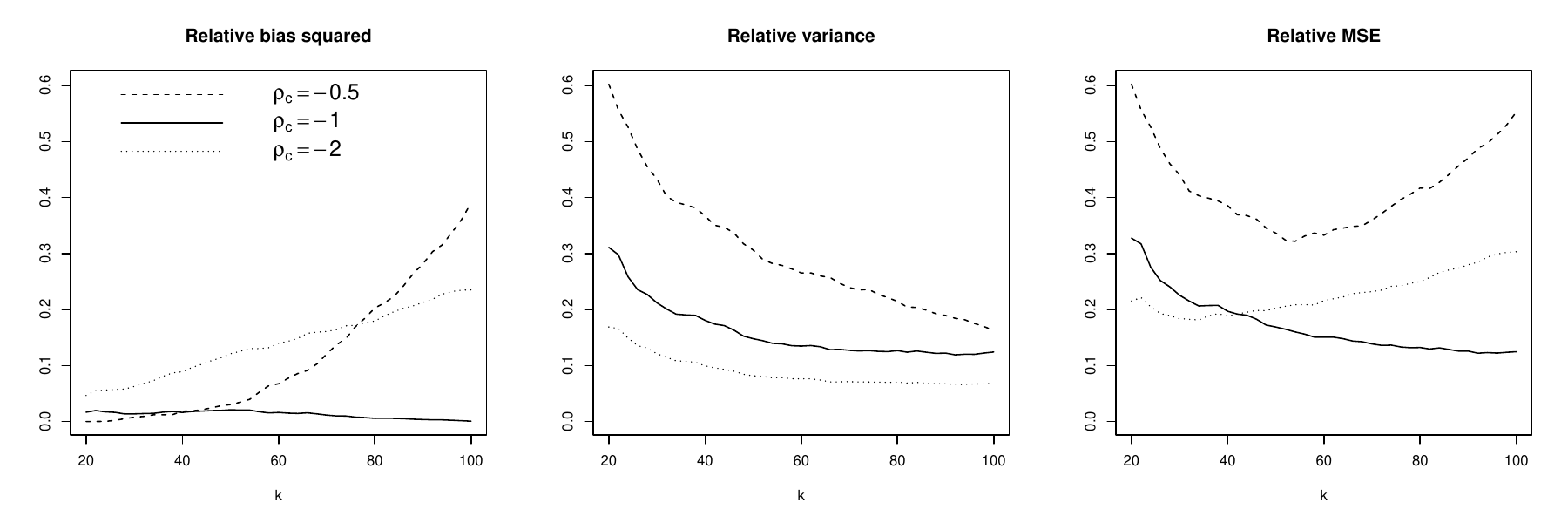}
    \caption{$\nu=4$, $\rho=-1/2$}
\end{subfigure}

%\vspace{0.5em}

\begin{subfigure}{0.83\textwidth}
    \centering
    \includegraphics[width=\linewidth]{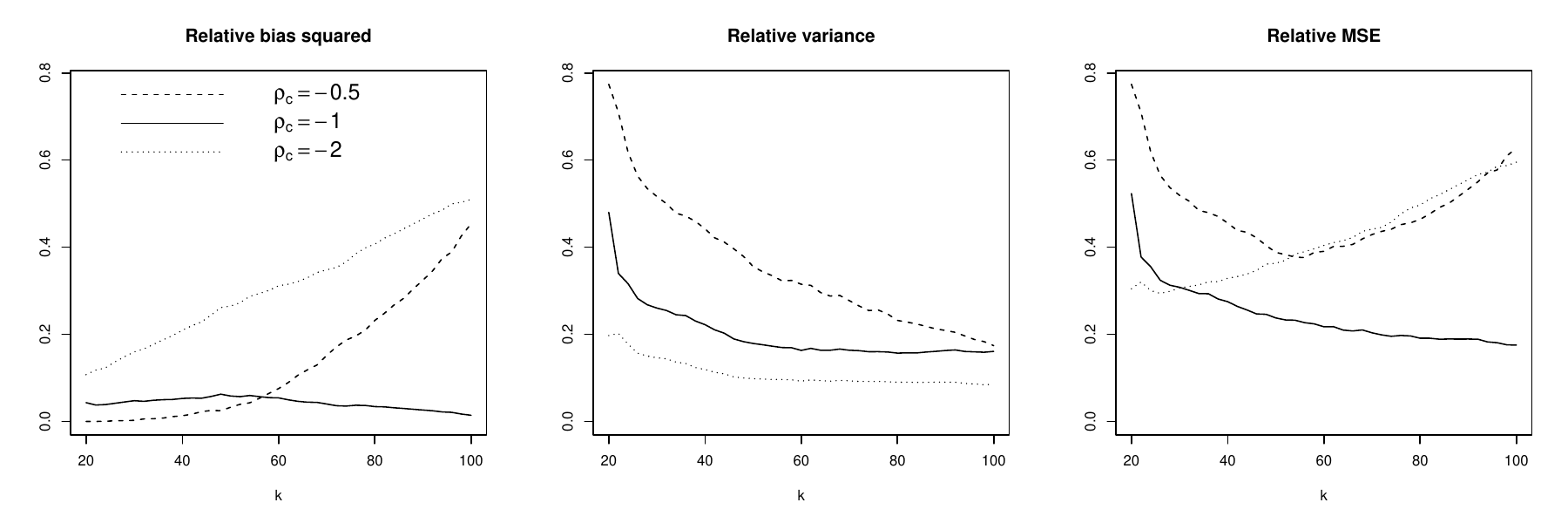}
    \caption{$\nu=5$, $\rho=-2/5$}
\end{subfigure}

\caption{Empirical relative bias squared, variance and MSE of the bias-corrected MELE of tail index $\gamma$ in~\eqref{qMELE} for three canonical values $\rho_c$ based on 1000 replications of random samples of size $n=500$ from Student t distributions with $\nu$ degrees of freedom, $\nu\in\{2,3,4,5\}$.}\label{fig:rc.t}
\end{figure}

\captionsetup[subfigure]{skip=2pt}

\begin{figure}[htbp]
\centering

\begin{subfigure}{0.83\textwidth}
    \centering
    \includegraphics[width=\linewidth]{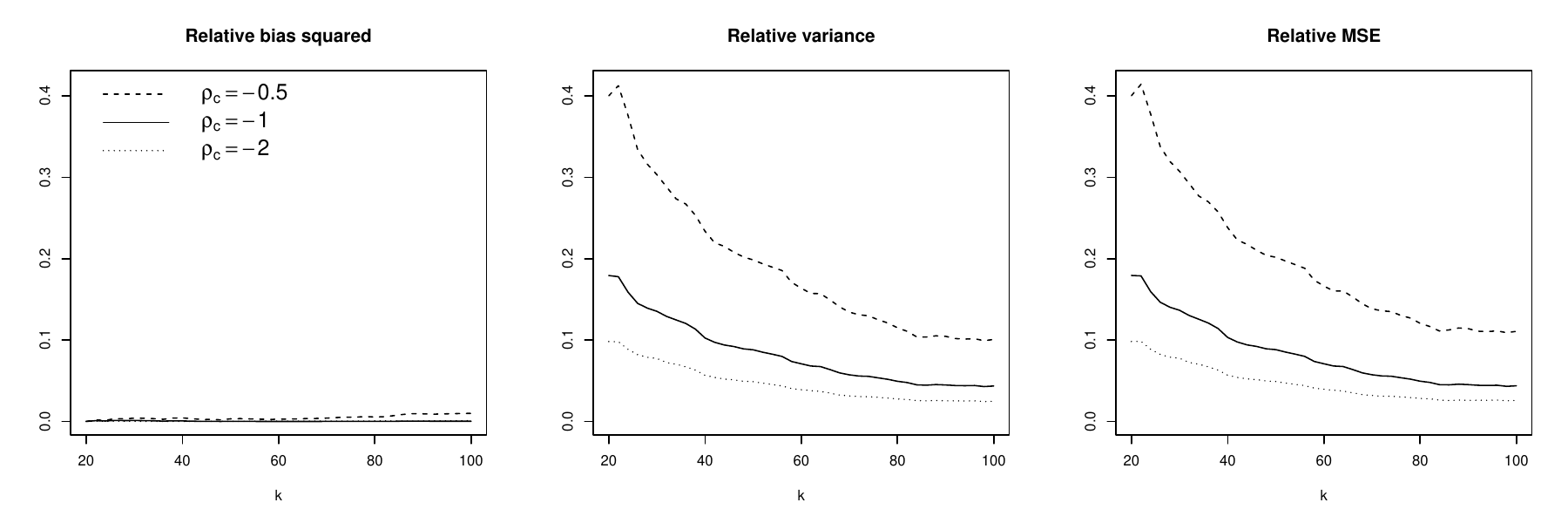}
    \caption{$\rho=-1$}
\end{subfigure}

%\vspace{0.5em}

\begin{subfigure}{0.83\textwidth}
    \centering
    \includegraphics[width=\linewidth]{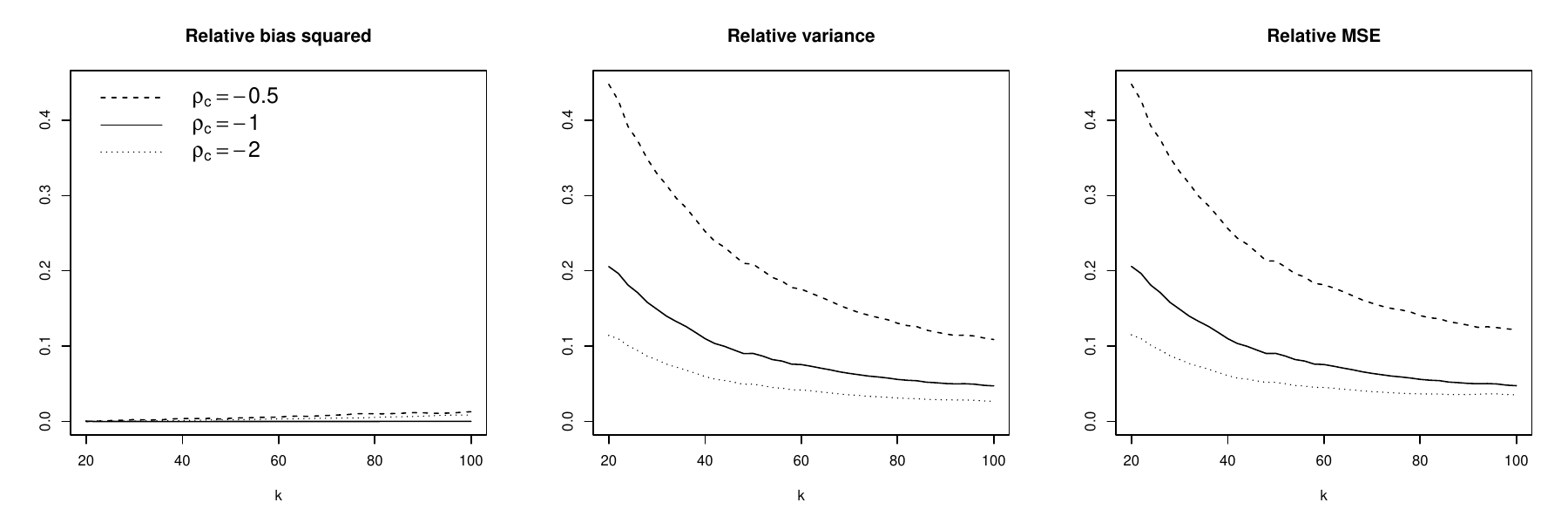}
    \caption{$\rho=-3/4$}
\end{subfigure}

%\vspace{0.5em}

\begin{subfigure}{0.83\textwidth}
    \centering
    \includegraphics[width=\linewidth]{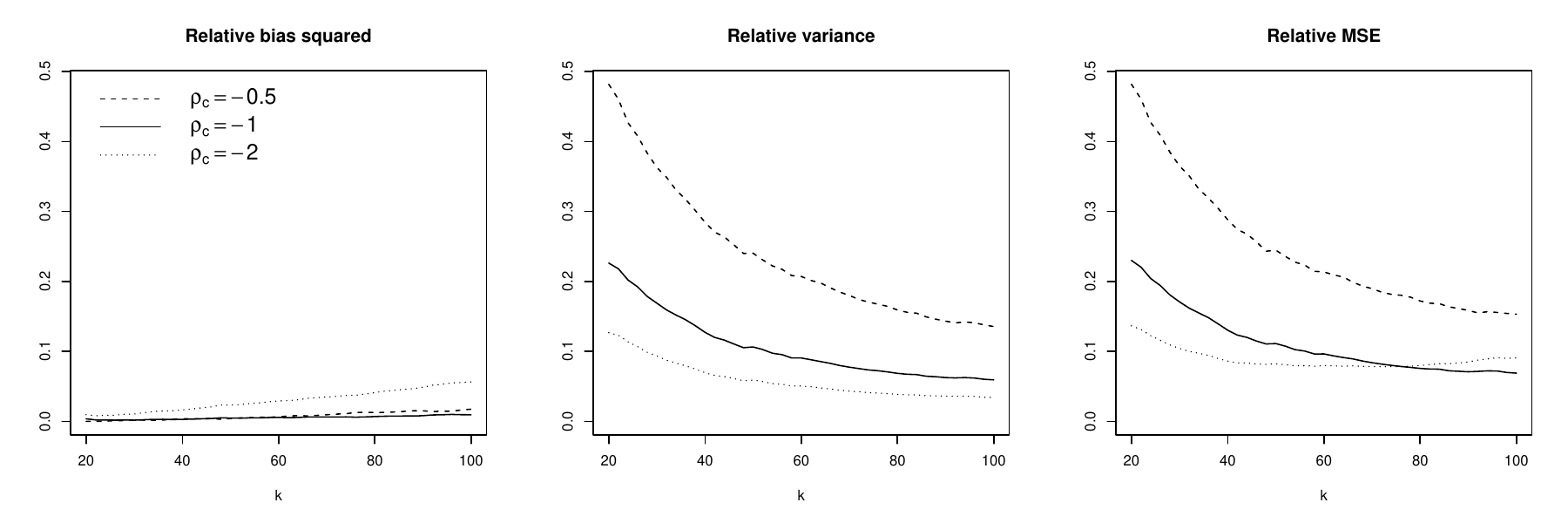}
    \caption{$\rho=-1/2$}
\end{subfigure}

%\vspace{0.5em}

\begin{subfigure}{0.83\textwidth}
    \centering
    \includegraphics[width=\linewidth]{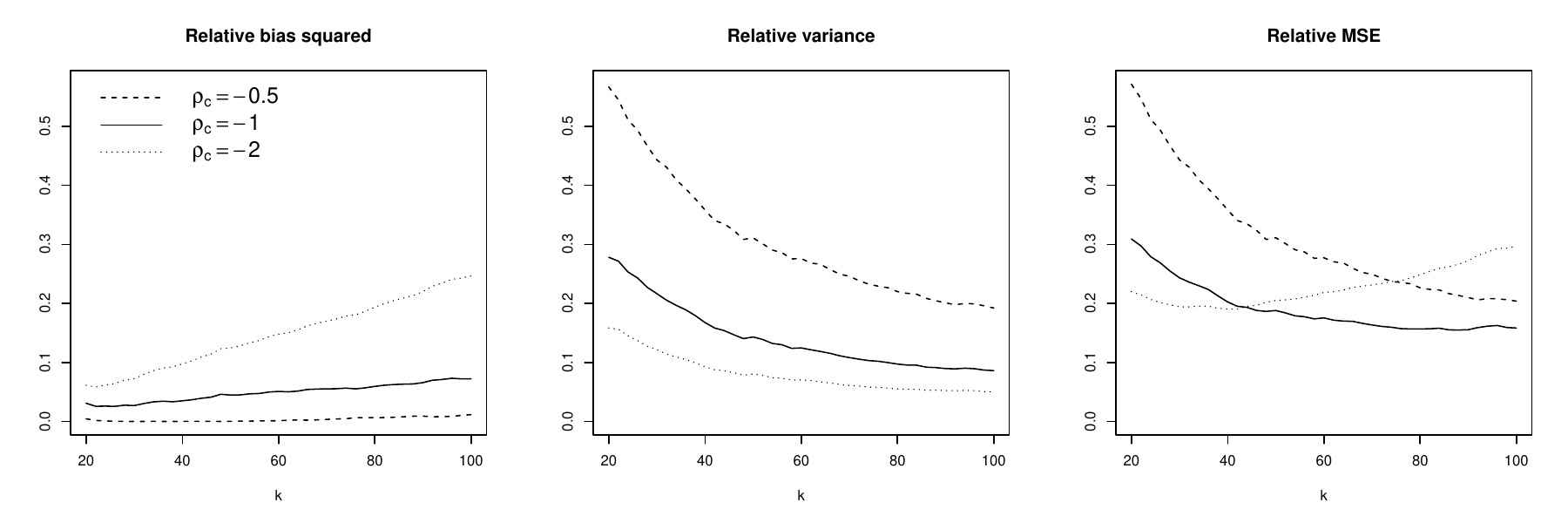}
    \caption{$\rho=-1/3$}
\end{subfigure}

\caption{Empirical relative bias squared, variance and MSE of the bias-corrected MELE of tail index $\gamma$ in~\eqref{qMELE} for three canonical values $\rho_c$ based on 1000 replications of random samples of size $n=500$ from Burr distributions with parameters $\tau=1/\l$ and $\l\in\{1,4/3,2,3\}$; second order parameter is $\rho=-1/\l$.}\label{fig:rc.burr}
\end{figure}

\begin{figure}[h]
    \centering
 \includegraphics[width=1\textwidth]{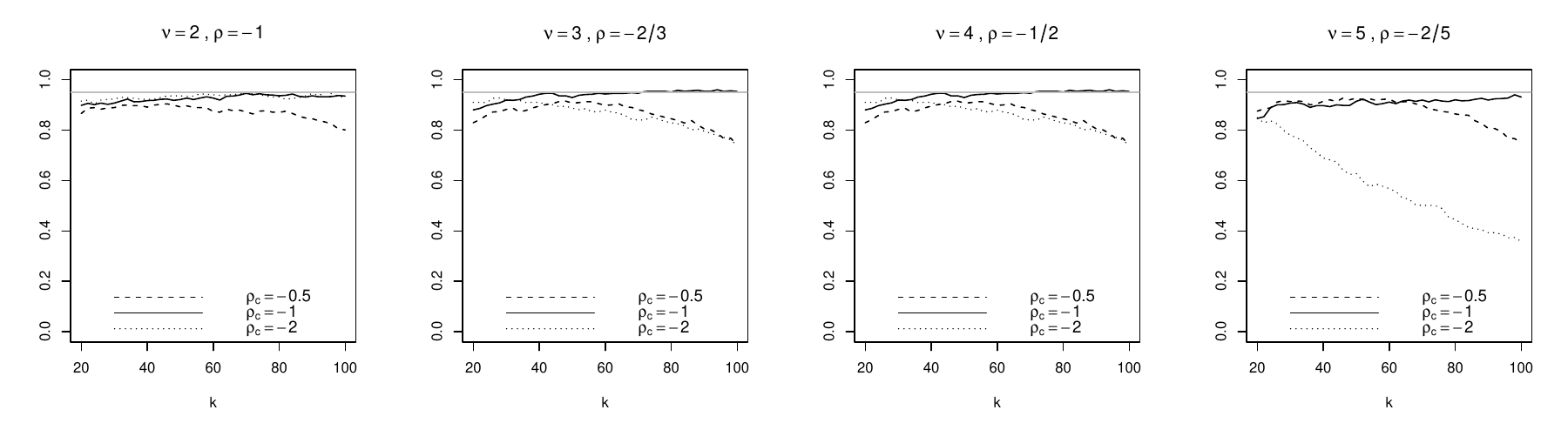}
         \caption{Empirical coverage probabilities for 95\% confidence intervals based on the bias-corrected MELE of tail index $\gamma$ in~\eqref{qMELE} for three canonical values $\rho_c$ based on 1000 replications of random samples of size $n=500$ from Student t distributions with $\nu$ degrees of freedom, $\nu\in\{2,3,4,5\}$.}
         \label{fig:rcc.t}
\end{figure}

\begin{figure}[h]
    \centering
 \includegraphics[width=1\textwidth]{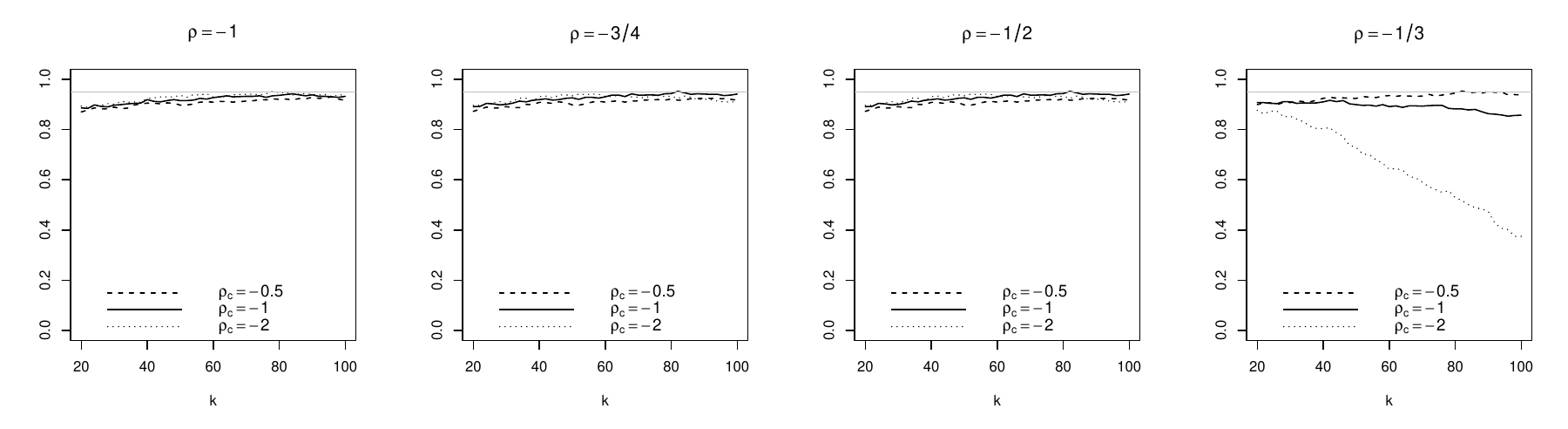}
         \caption{Empirical coverage probabilities for 95\% confidence intervals based on the bias-corrected MELE of tail index $\gamma$ in~\eqref{qMELE} for three canonical values $\rho_c$ based on 1000 replications of random samples of size $n=500$ from Burr distributions with parameters $\tau=1/\l$ and $\l\in\{1,4/3,2,3\}$; second order parameter is $\rho=-1/\l$.}
         \label{fig:rcc.burr}
\end{figure}

%\clearpage
\subsection{Comparison of tail index estimators and confidence interval constructions}

We next compare the proposed estimator of the tail index with its direct competitors: the Hill estimator in~\eqref{qhill} and the bias-corrected (parametric) MLE of \cite{Beirlant1999}. Based on the asymptotic behaviour of these three estimators, the Hill estimator has the smallest (asymptotic) variance and this is confirmed by the middle panel plots in Figures~\ref{fig:est.t} and~\ref{fig:est.burr}. Due to the choice of $\r$ values in the simulation settings, as expected, the Hill estimator shows bias in finite sample situations, which can be substantial for larger values of $\r$. The bias-corrected MLE exhibits the largest empirical variance\footnote{Since estimation of $\rho$ is inconsistent, \cite{Beirlant1999} set an upper bound of -0.5 in estimation of $\rho$ to stabilize estimation of the tail index. The variance of bias-corrected MLE depends on this upper bound.}, which can be attributed to estimation of an extra parameter, $\r$, but it does lead to the lowest bias in most cases, particularly when $\r\ge-0.5$. In comparison, the proposed estimator seems to trade some of the bias (due to $\r$ misspecification) for variance reduction which manifests in overall lowest MSE over moderate values of $k$ when $-1<\r\le -0.5$ and outperforms other method for $\r> -0.5$.

\captionsetup[subfigure]{skip=2pt}

\begin{figure}[htbp]
\centering

\begin{subfigure}{0.83\textwidth}
    \centering
    \includegraphics[width=\linewidth]{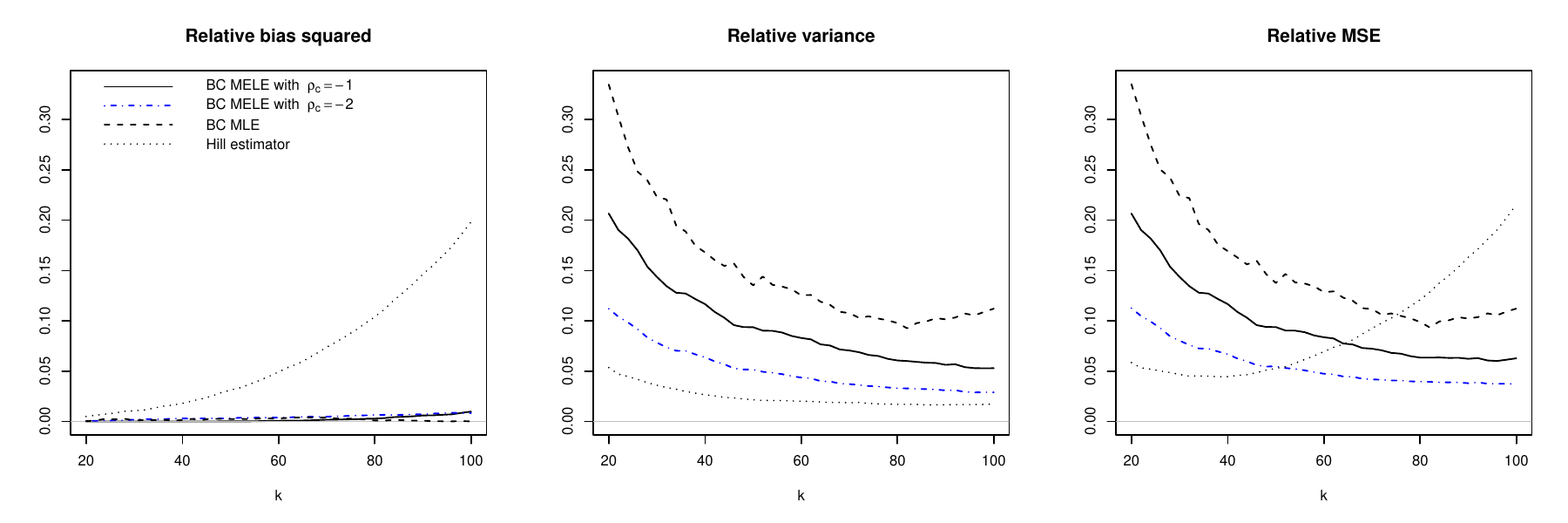}
    \caption{$\nu=2$, $\rho=-1$}
\end{subfigure}

%\vspace{0.5em}

\begin{subfigure}{0.83\textwidth}
    \centering
    \includegraphics[width=\linewidth]{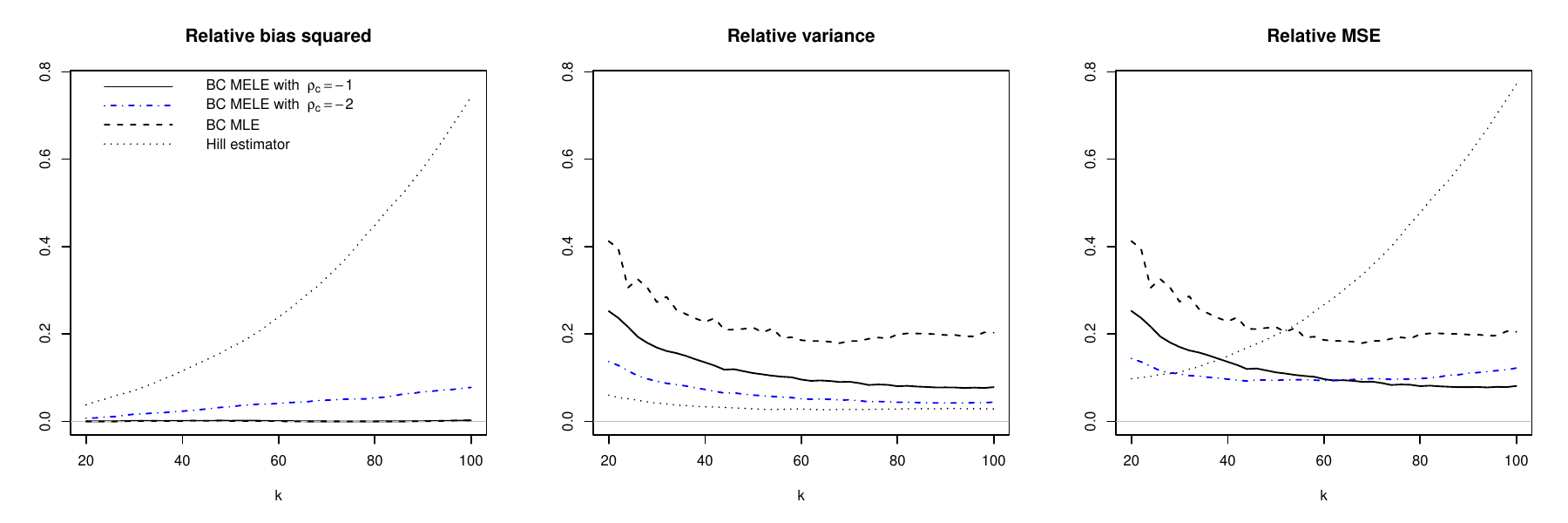}
    \caption{$\nu=3$, $\rho=-2/3$}
\end{subfigure}

%\vspace{0.5em}

\begin{subfigure}{0.83\textwidth}
    \centering
    \includegraphics[width=\linewidth]{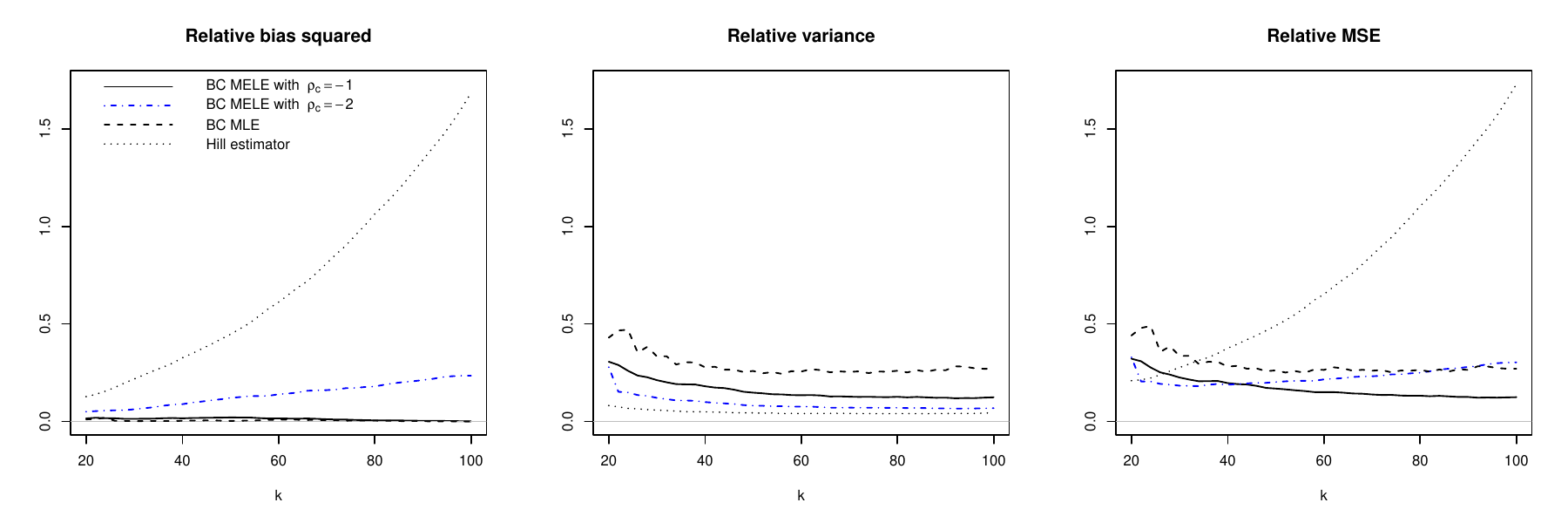}
    \caption{$\nu=4$, $\rho=-1/2$}
\end{subfigure}

%\vspace{0.5em}

\begin{subfigure}{0.83\textwidth}
    \centering
    \includegraphics[width=\linewidth]{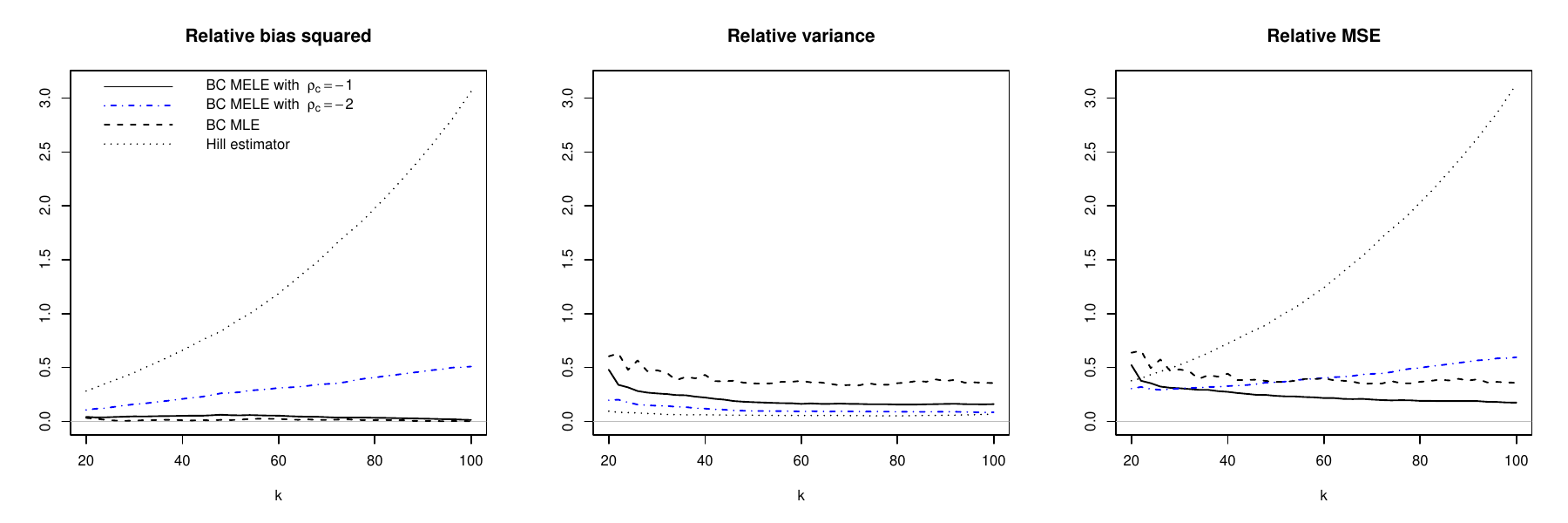}
    \caption{$\nu=5$, $\rho=-2/5$}
\end{subfigure}

\caption{Empirical relative bias squared, variance and MSE of the bias-corrected MELE of tail index $\gamma$ in~\eqref{qMELE} for two canonical values $\rho_c$, the bias-corrected MLE and the Hill estimator based on 1000 replications of random samples of size $n=500$ from Student t distributions with $\nu$ degrees of freedom, $\nu\in\{2,3,4,5\}$.}\label{fig:est.t}
\end{figure}

\captionsetup[subfigure]{skip=2pt}

\begin{figure}[htbp]
\centering

\begin{subfigure}{0.83\textwidth}
    \centering
    \includegraphics[width=\linewidth]{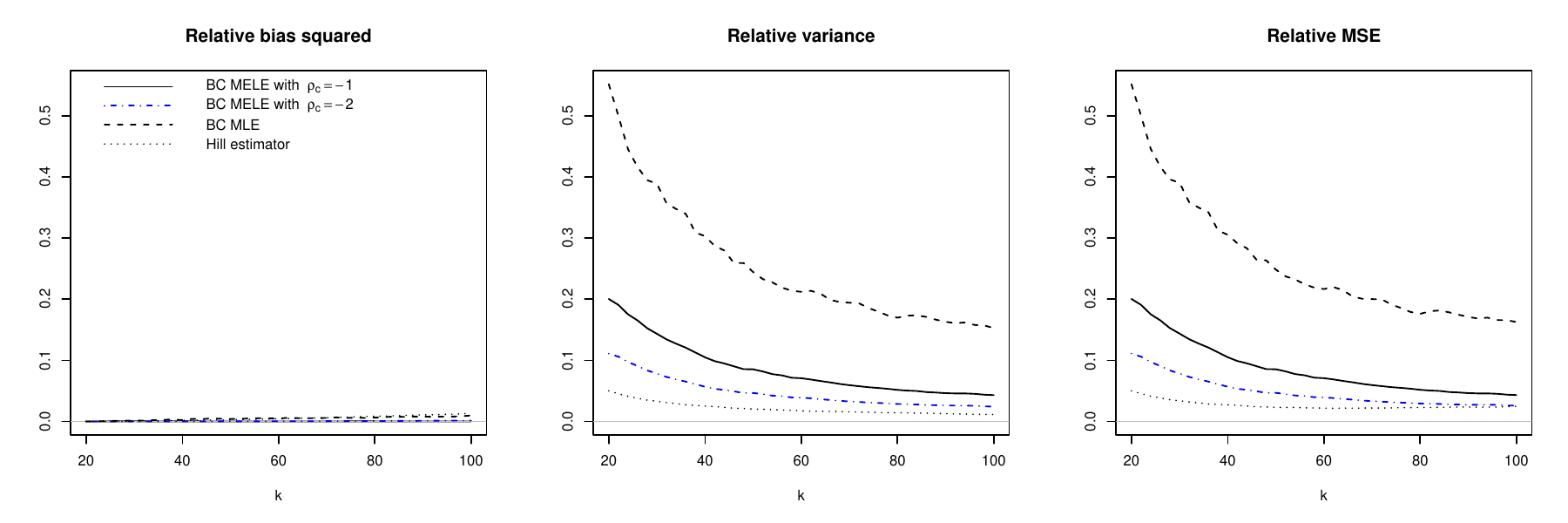}
    \caption{$\rho=-1$}
\end{subfigure}

%\vspace{0.5em}

\begin{subfigure}{0.83\textwidth}
    \centering
    \includegraphics[width=\linewidth]{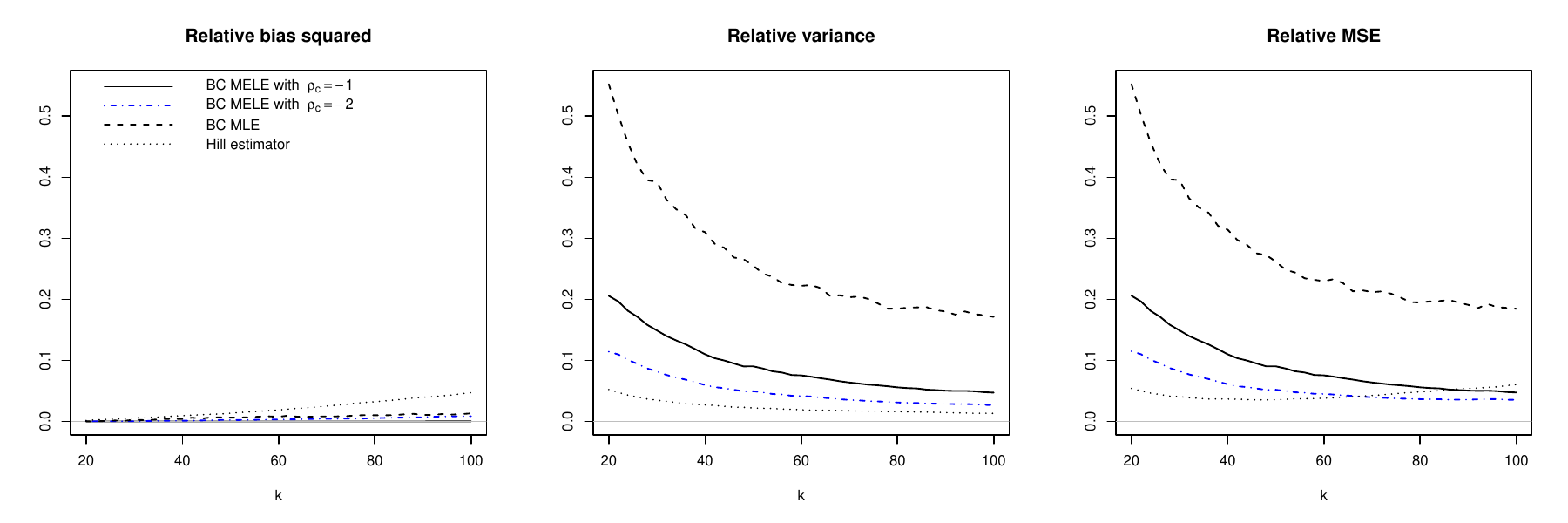}
    \caption{$\rho=-3/4$}
\end{subfigure}

%\vspace{0.5em}

\begin{subfigure}{0.83\textwidth}
    \centering
    \includegraphics[width=\linewidth]{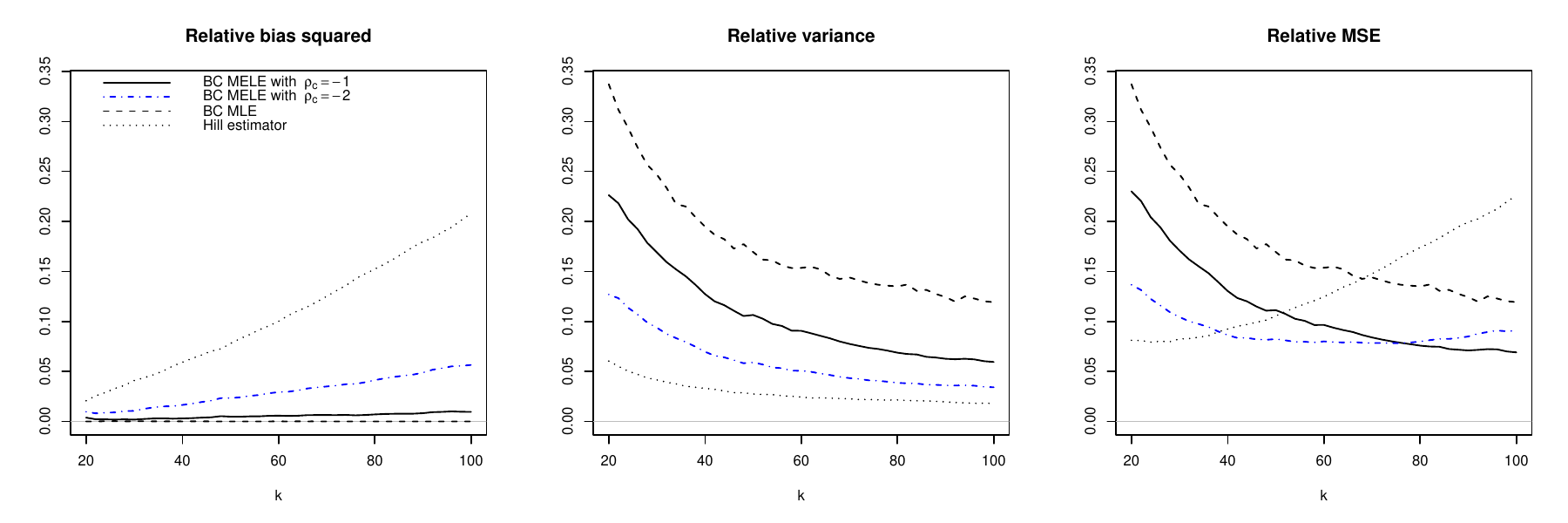}
    \caption{$\rho=-1/2$}
\end{subfigure}

%\vspace{0.5em}

\begin{subfigure}{0.83\textwidth}
    \centering
    \includegraphics[width=\linewidth]{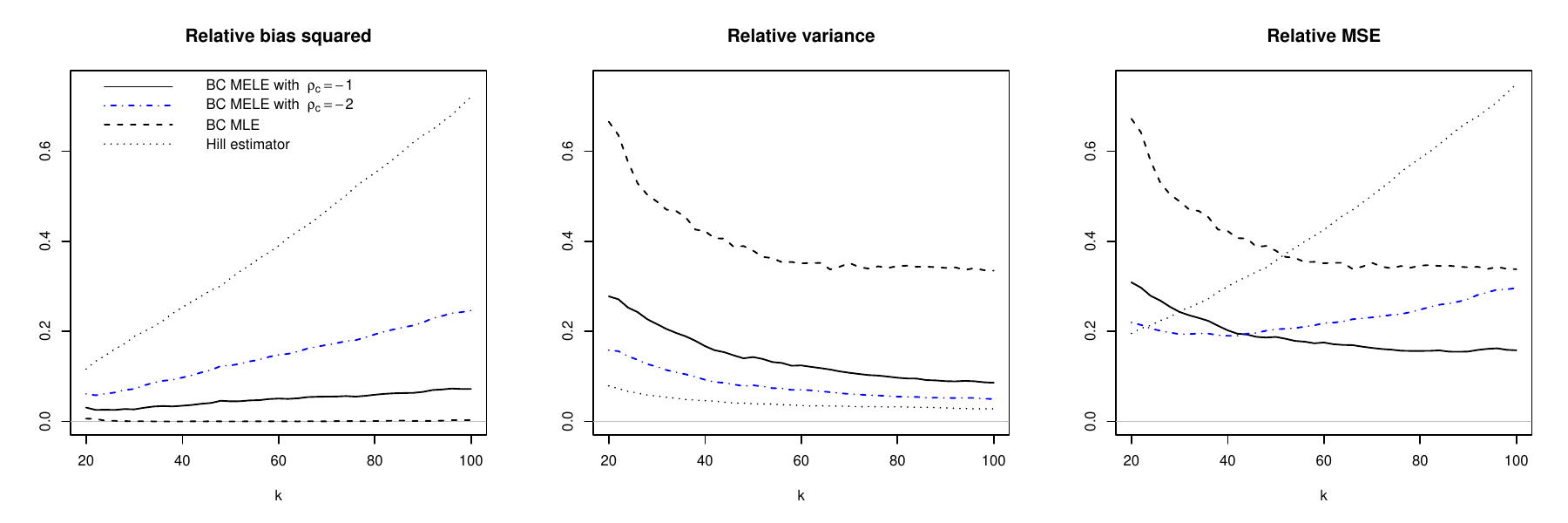}
    \caption{$\rho=-1/3$}
\end{subfigure}

\caption{Empirical relative bias squared, variance and MSE of the bias-corrected MELE of tail index $\gamma$ in~\eqref{qMELE} for two canonical values $\rho_c$, the bias-corrected MLE and the Hill estimator  based on 1000 replications of random samples of size $n=500$ from Burr distributions with parameters $\tau=1/\l$ and $\l\in\{1,4/3,2,3\}$; second order parameter is $\rho=-1/\l$.}\label{fig:est.burr}
\end{figure}

Finally, we examine the empirical coverage rates across the following four confidence interval constructions:
\begin{itemize}
\item the bias-corrected empirical likelihood confidence interval (BCEL), $\mathcal{I}_{BCEL}$ in \eqref{bel.ci} with $\r_c=-1$;
\item the parametric bias-corrected confidence interval (BCMLE) of \cite{Beirlant1999} with upper bound of -0.5 for optimization of $\rho$; 
\item the empirical likelihood-based confidence interval (EL),  $\mathcal{I}_{EL}$ in \eqref{ci.el};
\item the confidence interval based on asymptotic normality of the Hill estimator (Hill), $\mathcal{I}_{H}$ in \eqref{ci.hill}.
%\item The adjusted empirical likelihood confidence interval (AEL) of \cite{Li2019} with the theoretical value $a_n=19/12$ that achieves second order coverage precision. 
\end{itemize}
In Figures~\ref{fig:cvg.t} and~\ref{fig:cvg.burr} we show the empirical coverage rates for 95\% confidence intervals under the different constructions listed above. As already mentioned earlier, the coverage rates appear to be largely influenced by the bias of underlying tail index estimators. This explains substantial undercoverage for the Hill and EL confidence intervals when $\r$ is increasing away from $-1$ or $k$ moves towards a more moderate range. The estimation of the asymptotic variance of the tail index estimator also has an impact on the coverage rates. For the parametric bias-corrected MLE, lack of a consistent estimator of~$\rho$ is likely to lead to inaccurate estimation of the asymptotic variance. It is interesting to note that the parametric bias-corrected confidence interval often has coverage rates above the nominal level due to overestimation of the asymptotic variance, and coverage rates either above or below the nominal level for larger values of $k$. The proposed bias-corrected empirical likelihood confidence interval (with $\r_c=-1$) mostly maintains the nominal coverage for $-1\le\r\le-0.5$ and moderate values of~$k$, but does begin to show undercoverage when $\r>-0.5$ due to larger discrepancies between $\r$ and $\r_c$.

\begin{figure}[h]
    \centering
 \includegraphics[width=1\textwidth]{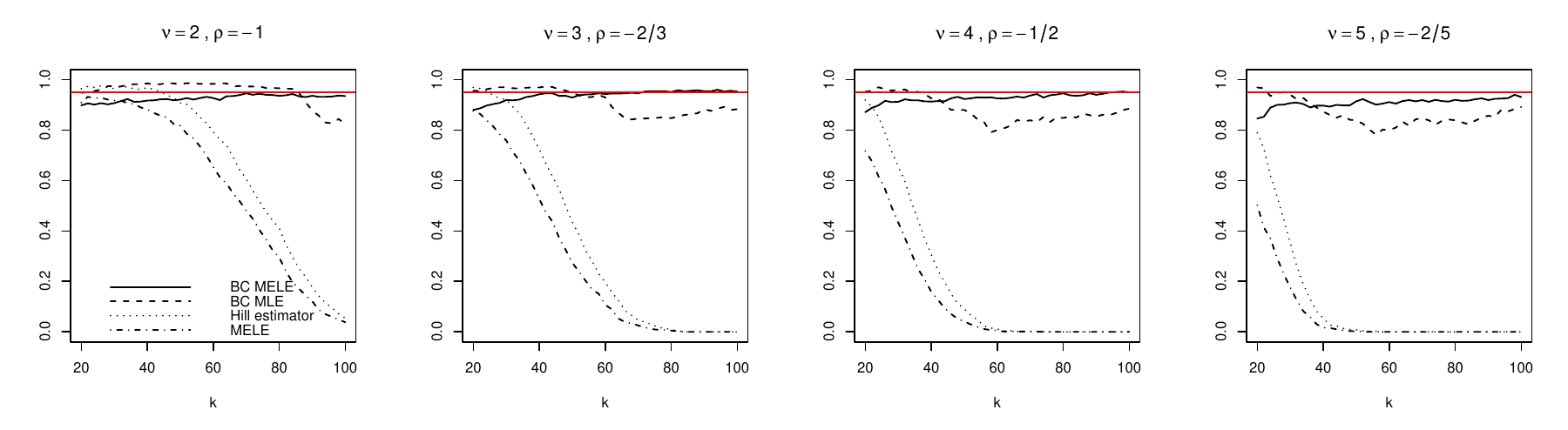}
         \caption{Empirical coverage probabilities for 95\% confidence intervals under different constructions including the bias-correct EL interval in~\eqref{bel.ci} with canonical value $\rho_c=-1$ (denoted "BC MELE") and  confidence intervals based on the bias-corrected MLE of \citet{Beirlant1999} (denoted "BC MLE"), the Hill estimator in \eqref{ci.hill} and the empirical likelihood in~\eqref{ci.el} (denoted "MELE"). The results are based on 1000 replications of random samples of size $n=500$ from Student t distributions with $\nu$ degrees of freedom, $\nu\in\{2,3,4,5\}$.}
         \label{fig:cvg.t}
\end{figure}

\begin{figure}[h]
    \centering
 \includegraphics[width=1\textwidth]{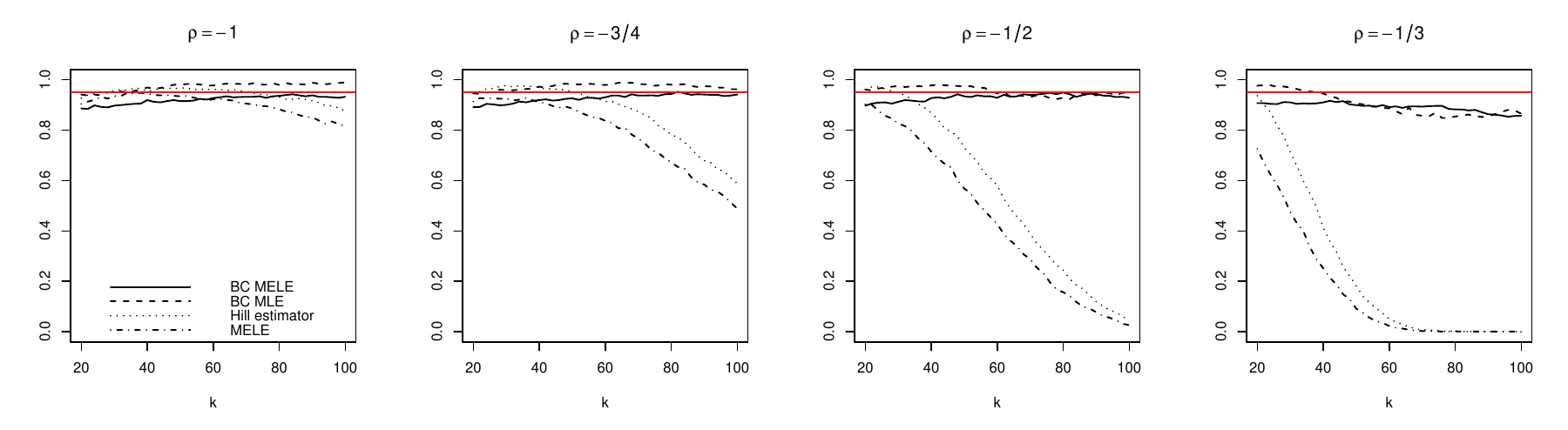}
         \caption{Empirical coverage probabilities for 95\% confidence intervals under different constructions including the bias-correct EL interval in~\eqref{bel.ci} with canonical value $\rho_c=-1$ (denoted "BC MELE") and  confidence intervals based on the bias-corrected MLE of \citet{Beirlant1999} (denoted "BC MLE"), the Hill estimator in \eqref{ci.hill} and the empirical likelihood in~\eqref{ci.el} (denoted "MELE"). The results are based on 1000 replications of random samples of size $n=500$ from Burr distributions with parameters $\tau=1/\l$ and $\l\in\{1,4/3,2,3\}$; second order parameter is $\rho=-1/\l$.}
         \label{fig:cvg.burr}
\end{figure}

\clearpage
\section{Application}\label{s.appl}
In this section, we illustrate the proposed estimator of the tail index on a real life dataset. We consider the Internet trace data downloaded from the Boston University's Computer Science Department in December 1994; the data can be accessed at the following link: \href{https://ita.ee.lbl.gov/html/contrib/BU-Web-Client.html}{https://ita.ee.lbl.gov/html/contrib/BU-Web-Client.html}.
 During the period of study, a total of 21,108 files were requested, and the size of each requested file (in bytes) was recorded. A large number of the files have zero size, and many files in the dataset contain duplicates. After removing empty and duplicate files, the dataset contains $n=2,785$ distinct files. Figure~\ref{fig:loglogsurv} shows a log-log survival plot with the ordered observations (on the horizontal axis) plotted against empirical estimates of their exceedance probability. Under the regular variation model in~\eqref{reg.var}, we expect the tail of the logarithm of the survival function to be approximately linear in $\log x$ with slope $-1/\g$. For this dataset, the tail does appear to taper off linearly, thus justifying the regular variation assumption.

\begin{figure}[h]
    \centering
     \includegraphics[width=0.5\textwidth]{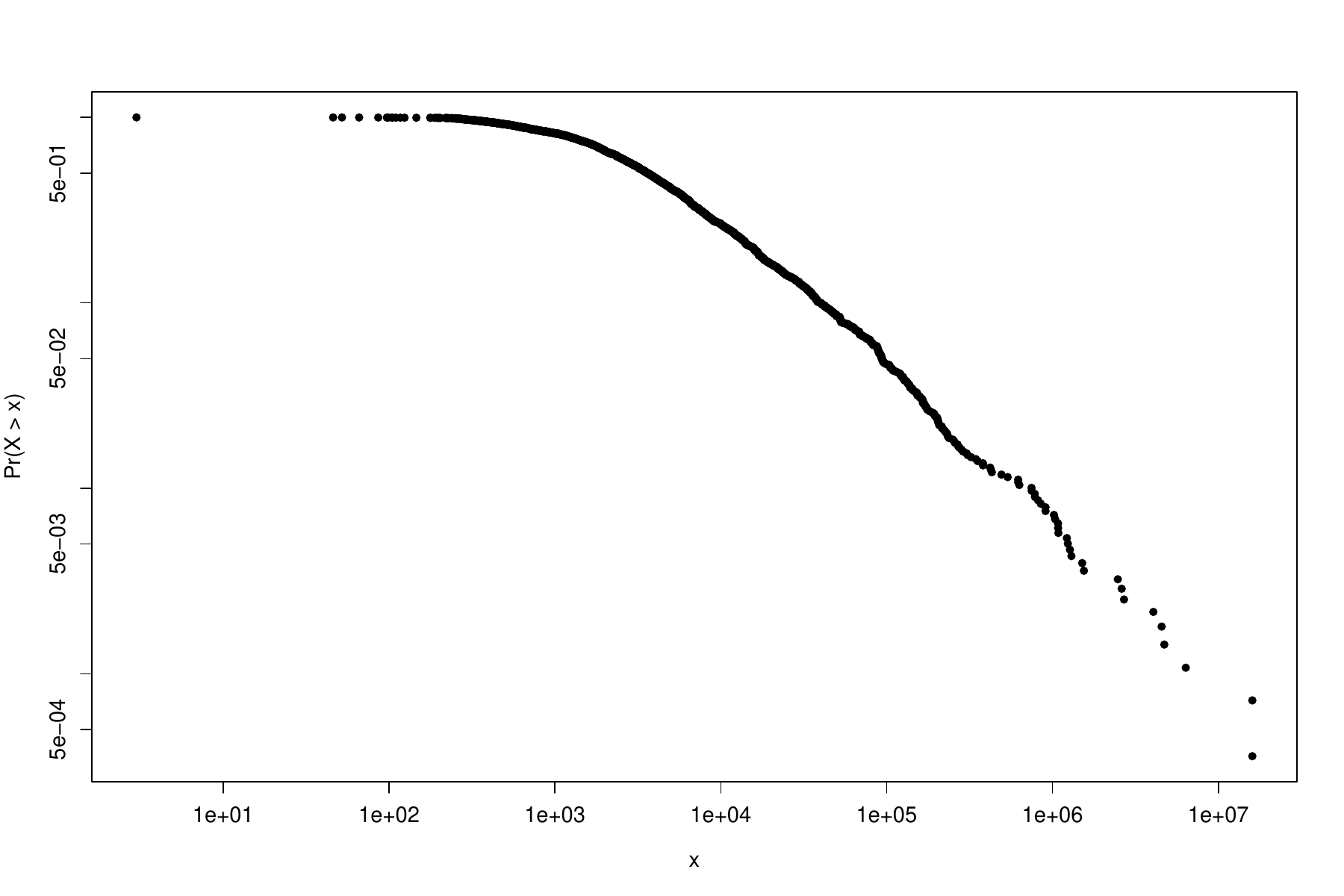}
    \caption{Log-log survival plot for the file size dataset based on 2,785 distinct files.}
\label{fig:loglogsurv}
\end{figure}

An important step in the application of many extreme value methods is deciding where the tail begins so an asymptotic model can be relied on to provide an accurate approximation to the tail of the underlying distribution. In our set-up, this requires the choice of sample fraction~$k$. For the classical Hill estimator, a standard approach to choosing $k$ is to explore tail index estimates over a range of values of $k$ and select a value of $k$ where the estimates have plateaued. A similar technique can be adopted to select~$k$ in the proposed bias-corrected MELE in~\eqref{qMELE}. Figure~\ref{fig:hillplot} illustrates the two plots based on the Hill estimator and bias-corrected MELE with $\r_c=-1$. Both plots exhibit a fluctuating behaviour for smaller values of~$k$. The Hill plot seems to settle only from about $k=470$ to $k=680$ (indicated by vertical dashed lines), with the average value of the estimates in this region giving $\widehat\g_H \approx 1.30$. However, these values of~$k$ are rather large relative to the sample size here, corresponding to the range from 17\% to 24\% of the sample, and thus are likely to lead to an estimate that is biased. The bias-corrected MELE plot has a stable region in a range $600 \le k\le 750$. This range of $k$ values is less of an issue for the bias-corrected estimator as it incorporates second-order terms in the asymptotic model, which allows to select larger values of $k$, a fact also supported by the simulation studies. Taking the average over estimates in the above region leads to $\widehat\g_E \approx 1.10$. 

To validate these estimates, we also consider the log-tail probability plot for the largest 10\% of the sample, or 278 observations in our case; see Figure~\ref{fig:diag}. Superimposed on this plot are lines obtained by fitting the intercept using least squares while keeping the slope fixed at $-1/\widehat\g$ for each of the two tail index estimates. On the basis of this plot, we find that the tail index estimate resulted from the proposed bias-corrected MELE provides a better fit to the tail observations, more accurately capturing the rate of decay compared to the Hill estimator based on the plateau region of the Hill plot. The residual sums of squares are 0.44 and 0.71 for the bias-corrected MELE and Hill estimator, respectively. As a final note on the comparison of the two estimators, we remark that while the Hill estimator leads to a more narrow confidence band (see Figure~\ref{fig:hillplot}), this band likely underestimates its sampling variability as was captured by the coverage rates in the simulation studies.

\begin{figure}[h]
    \centering
     \includegraphics[width=0.9\textwidth]{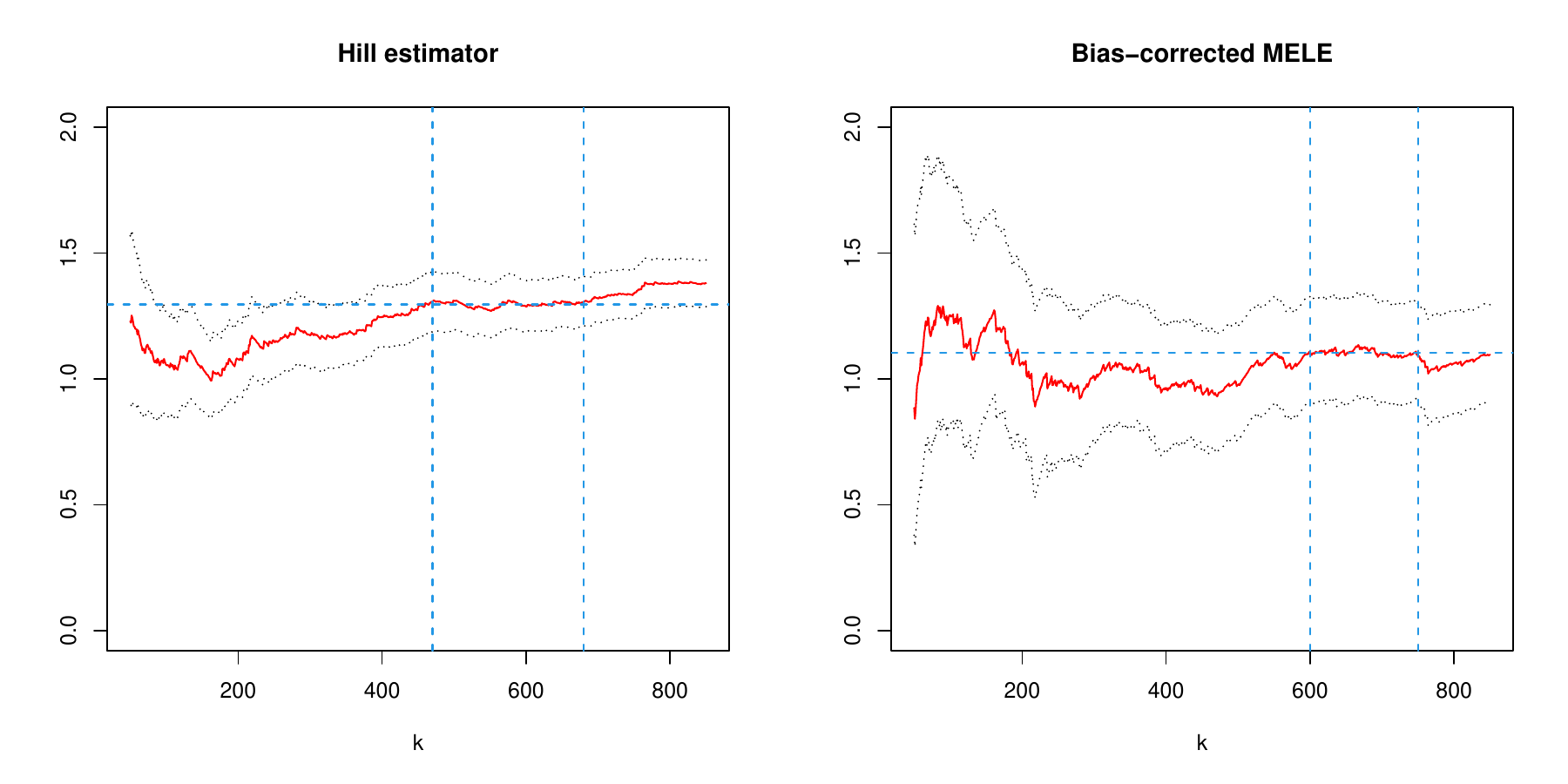}
    \caption{Plots of the tail index estimates (solid red curves) and corresponding 95\% confidence intervals (dotted curves) as a function of sample fraction~$k$ based on the Hill estimator in~\eqref{qhill} (left panel) and bias-corrected maximum empirical likelihood estimator in~\eqref{qMELE} for the file size dataset.}
\label{fig:hillplot}
\end{figure}

\begin{figure}[h]
    \centering
     \includegraphics[width=0.5\textwidth]{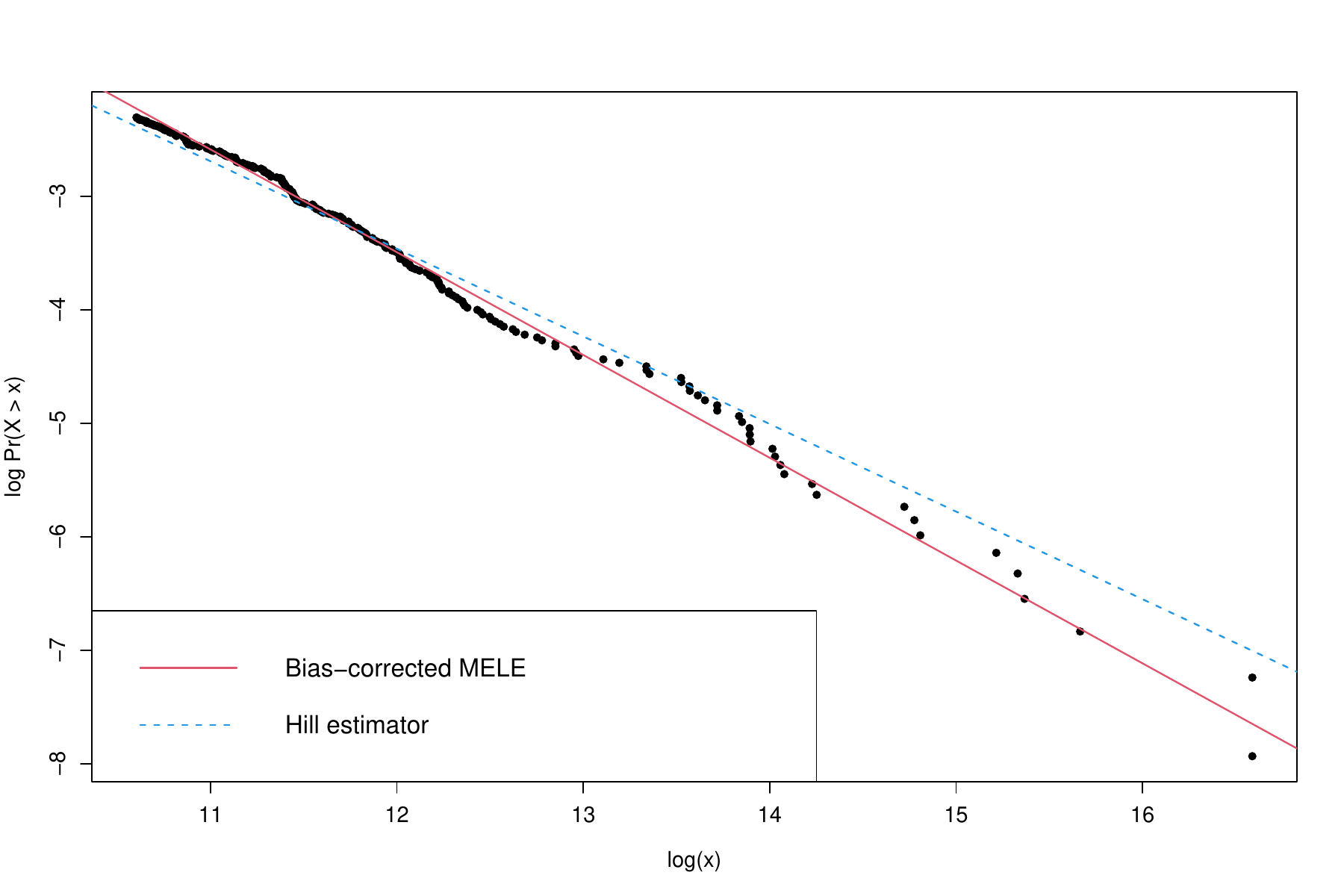}
    \caption{Log-tail probability plot based the largest 278 observations (10\% of the sample) in the file size dataset. The lines indicate least-squares fits to these tail observations with slopes held fixed at $-1/\widehat\g_H$ and $-1/\widehat\g_E$ for the Hill estimate and bias-corrected MELE, respectively.}
\label{fig:diag}
\end{figure}

\section{Conclusion}\label{s.concl}

The paper makes a contribution to the literature on tail index estimation for heavy-tailed data, an important yet often challenging problem arising in a wide range of applications.  We combine the idea of bias correction from \cite{Beirlant1999} with the empirical likelihood approach to propose a bias-corrected estimator for the tail index based on
the partial derivatives of the least squares objective function. Instead of jointly estimating the three parameters as in \cite{Beirlant1999}, we replace the second order parameter $\rho$ by a canonical choice $\rho_c$. This step is motivated by the fact that estimation of~$\rho$ through joint optimization of the parametric likelihood function is generally numerically unstable and inconsistent, often leading to inconsistent estimation of the asymptotic variance of the MLE of~$\g$ and, as a consequence, poor coverage rates of corresponding confidence intervals. 
The proposed bias-corrected estimator under suitable normalization has an asymptotically normal distribution with asymptotic variance depending on the choice of $\r_c$. The resulting bias-corrected EL ratio statistic has a $\chi^2$ limiting distribution, which guarantees accurate coverage when sample size~$n$ is large. The resulting bias-corrected EL confidence interval is data driven and does not require explicit estimation of the asymptotic variance.
While the use of a canonical value $\r_c$ may introduce a bias in finite samples when $\r$ and $\r_c$ are far apart, the asymptotic results as well as the simulation studies show that, when setting $\r_c = -1$, a reduction in variance often outweighs the effect of bias and results in lower MSE when comparing the proposed estimator with the bias-corrected parametric MLE. The proposed estimator also offers a competitive alternative to the Hill estimator in cases of slow convergence when $\r>-1$.

Theoretical considerations and simulation studies suggest $\r_c=-1$ as a preferred choice across most of the considered data generating processes both in terms of point estimation metrics and confidence interval coverage rates. The choice of sample fraction~$k$ plays an important role in the performance of tail index estimators discussed above. In contrast to the Hill estimator for which bias increases rapidly with $k$, the proposed bias-corrected estimator tends to work well across a wider range of $k$ values, thus allowing to select a larger value of $k$ for a reduced variance but without sacrificing much in terms of bias. Empirical coverage rates also tend to improve for moderate values of $k$.

As a final remark, we note that the proposed estimator continues to share some of the drawbacks of the Hill estimator, in particular, lack of location-invariance due to reliance on log-spacings. While the use of empirical likelihood generally leads to range-respecting confidence intervals, the proposed interval in~\eqref{bel.ci} does not have this property due to the form of the objective function which requires estimation of the rate function $b_{n,k}$. In practice, this would be an issue only in cases of lighter-tailed data with small values of $\g$.

\section*{Acknowledgements} The authors acknowledge financial support of the Natural Sciences and Engineering Research Council of Canada. 

%\clearpage
\bibliographystyle{apalike}
\bibliography{bias_corrected_el}

\appendix
\def\thesection{Appendix \Alph{section}}
\def\thesubsection{\Alph{section}.\arabic{subsection}}
\def\thetheorem{\Alph{section}.\arabic{theorem}}
\setcounter{equation}{0}
\renewcommand{\theequation}{A.\arabic{equation}}

\section{Proofs}

\subsection{Proof of Theorem~\ref{asymptotic.normality}}\label{sA1}

\begin{proof}
In the proof we will use $\thb_0=(\g_0,b_{n,k})$ and $\r_0$ to denote the true values of the model parameters.
Let 
\[Z_j = \left[\g_0 + b_{n,k} \Big(\frac{j}{k+1}\Big)^{-\r_0}\right] E_j,\qquad j=1,\ldots,k,
\]
where $E_1,\ldots,E_k$ are i.i.d. standard exponential random variables, and set
\[\sigma_j^2 := Var(Z_j) = \left\{\gamma_0 + b_{n,k}\left(\frac{j}{k+1}\right)^{-\rho_0}\right\}^2,\qquad j=1,\ldots,k.
\]
According to Theorem~2.1 of \cite{Beirlant1999}, there exist random variables $\beta_j$ such that
\[\max_{1\leq j \leq k}|Y_j - Z_j - \beta_j| = o_p(b_{n,k})
\]
with $\sum_{j=1}^k |\b_j / j| = o_p(b_{n,k}\log k).$

To determine the order of $\sum_{j=1}^k g_{1j}(Y_j;\thb_0)$, we write
\begin{align*}
\sum_{j=1}^k g_{1j}(Y_j;\thb_0) &= \sum_{j=1}^k \left[Y_j - \gamma_0 - b_{n,k}\left(\frac{j}{k+1}\right)^{-\rho_c}\right]
\\&= \sum_{j=1}^k \left[Z_j - \gamma_0 - b_{n,k}\left(\frac{j}{k+1}\right)^{-\rho_0}\right]\\
& \quad + b_{n,k}\sum_{j=1}^k \left[\left(\frac{j}{k+1}\right)^{-\rho_0} - \left(\frac{j}{k+1}\right)^{-\rho_c}\right] + \sum_{j=1}^k (Y_j - Z_j - \beta_j) + \sum_{j=1}^k \beta_j.
\end{align*}

Under the condition that  $k^{1/2 + \delta} b_{n,k} = O(1)$, we obtain the following bound on the second term in the sum above:
\[\left|b_{n,k}\sum_{j=1}^k \left[\left(\frac{j}{k+1}\right)^{-\rho_0} - \left(\frac{j}{k+1}\right)^{-\rho_c}\right] \right| \leq  2kb_{n,k} = o(\sqrt{k}).
\]
The order of the third term is
\[\left |\sum_{j=1}^k Y_j - Z_j - \beta_j\right | \leq k\max_{1\leq j\leq k}|Y_j-Z_j-\beta_j| = o_p(kb_{n,k}) = o_p(\sqrt k),
\]
and the order of the fourth term is
\[\left|\sum_{j=1}^k \beta_j\right| \leq \sum_{j=1}^k |\b_j| \leq k\sum_{j=1}^{k}\left|\frac{\b_j}{j}\right| = o_p(kb_{n,k}\log k) = o_p(\sqrt{k}).
\]
Combining the above arguments, we have 
\[\sum_{j=1}^k g_{1j}(Y_j;\thb_0) = \sum_{j=1}^k \left[Z_j - \gamma_0 - b_{n,k}\left(\frac{j}{k+1}\right)^{-\rho_0}\right] + o_p(\sqrt k).
\]
Note the following relationship between $g_{1j}(Y_j;\thb)$ and $g_{2j}(Y_j;\thb)$:
\[g_{2j}(Y_j;\thb) =  g_{1j}(Y_j;\thb)\left(\frac{j}{k+1}\right)^{-\rho_c}.
\]
By the integral approximation,
\[\lim_{k\to\infty}\sum_{j=1}^k \left(\frac{j}{k+1}\right)^{-\rho_c} = \frac{1}{1-\rho_c}.
\]
Hence,
\[\sum_{j=1}^k g_{2j}(Y_j;\thb_0) = 
\sum_{j=1}^k \left[Z_j - \gamma_0 - b_{n,k}\left(\frac{j}{k+1}\right)^{-\rho_0}\right]\left(\frac{j}{k+1}\right)^{-\rho_c} + o_p(\sqrt{k})
\]
so that
\[\sum_{j=1}^k \gb_j(Y_j;\thb_0) = 
\sum_{j=1}^k
\begin{pmatrix}
g_{1j}(Y_j;\thb_0) \\
g_{2j}(Y_j;\thb_0) \\
\end{pmatrix} = \sum_{j=1}^k
\begin{pmatrix}
Z_j - \gamma_0 - b_{n,k}\left(\frac{j}{k+1}\right)^{-\rho_0} \\
\left[Z_j - \gamma_0 - b_{n,k}\left(\frac{j}{k+1}\right)^{-\rho_0}\right]\left(\frac{j}{k+1}\right)^{-\rho_c}
\end{pmatrix} + o_p(\sqrt k).
\]

Note
\begin{align*}
v_j &:= Var\begin{pmatrix}
Z_j - \gamma_0 - b_{n,k}\left(\frac{j}{k+1}\right)^{-\rho_0} \\
\left[Z_j - \gamma_0 - b_{n,k}\left(\frac{j}{k+1}\right)^{-\rho_0}\right]\left(\frac{j}{k+1}\right)^{-\rho_c}
\end{pmatrix}
\\& = \left[\gamma_0 + b_{n,k}\left(\frac{j}{k+1}\right)^{-\rho_0}\right]^2 \begin{pmatrix}
1 & \left(\frac{j}{k+1}\right)^{-\rho_c} \\
\left(\frac{j}{k+1}\right)^{-\rho_c} & \left(\frac{j}{k+1}\right)^{-2\rho_c} 
\end{pmatrix}
\end{align*}
and
\[V_k := (1/k)\sum_{j=1}^k v_j = \frac{1}{k}\sum_{j=1}^k \left[\gamma_0 + b_{n,k}\left(\frac{j}{k+1}\right)^{-\rho_0}\right]^2  
\begin{pmatrix}
1 & \left(\frac{j}{k+1}\right)^{-\rho_c} \\
\left(\frac{j}{k+1}\right)^{-\rho_c} & \left(\frac{j}{k+1}\right)^{-2\rho_c} 
\end{pmatrix}.
\]
This leads to an upper bound:
\begin{align*}
V_k &\leq \frac{1}{k}\sum_{j=1}^k (\gamma_0 + |b_{n,k}|)^2 \begin{pmatrix}
1 & \left(\frac{j}{k+1}\right)^{-\rho_c} \\
\left(\frac{j}{k+1}\right)^{-\rho_c} & \left(\frac{j}{k+1}\right)^{-2\rho_c} 
\end{pmatrix}
= (\gamma_0 + |b_{n,k}|)^2 \frac{1}{k}\sum_{j=1}^k
\begin{pmatrix}
1 & \left(\frac{j}{k+1}\right)^{-\rho_c} \\
\left(\frac{j}{k+1}\right)^{-\rho_c} & \left(\frac{j}{k+1}\right)^{-2\rho_c} 
\end{pmatrix},
\end{align*}
where the inequality is interpreted component-wise.
By the integral approximation,
\[\frac{1}{k}\sum_{j=1}^k
\begin{pmatrix}
1 & \left(\frac{j}{k+1}\right)^{-\rho_c} \\
\left(\frac{j}{k+1}\right)^{-\rho_c} & \left(\frac{j}{k+1}\right)^{-2\rho_c} 
\end{pmatrix} \to 
\begin{pmatrix}
1 & \frac{1}{1-\rho_c}\\
\frac{1}{1-\rho_c} & \frac{1}{1-2\rho_c} 
\end{pmatrix},\qquad k\to\infty.
\]
Hence,
\[\limsup_{k\to\infty} V_k \leq \gamma_0 ^2 \begin{pmatrix}
1 & \frac{1}{1-\rho_c}\\
\frac{1}{1-\rho_c} & \frac{1}{1-2\rho_c} 
\end{pmatrix}.
\]
Similarly,
\[\liminf_{k\to\infty} V_k \geq \gamma_0^2 \begin{pmatrix}
1 & \frac{1}{1-\rho_c}\\
\frac{1}{1-\rho_c} & \frac{1}{1-2\rho_c} 
\end{pmatrix}.
\]
Therefore, 
\[\lim_{k\to \infty} V_k = \limsup_{k\to\infty} V_k = \liminf_{k\to\infty} V_k = \gamma_0^2 \begin{pmatrix}
1 & \frac{1}{1-\rho_c}\\
\frac{1}{1-\rho_c} & \frac{1}{1-2\rho_c}
\end{pmatrix} =: S_{11}.
\]

Next we verify the conditions for the Lindberg Feller's Central Limit Theorem (CLT). We have:
\[\max_{1\leq j \leq k}\frac{(v_j)_{1,1}}{\sum_{j=1}^k (v_j)_{1,1}} \leq 
\frac{(\gamma_0 + |b_{n,k}|)^2}{k(\gamma_0 - |b_{n,k}|)^2} \to 0,\qquad k \to \infty.
\]
Furthermore,
\[\max_{1\leq j \leq k}\frac{(v_j)_{1,2}}{\sum_{j=1}^k (v_j)_{1,2}} \leq 
\frac{(\gamma_0 + |b_{n,k}|)^2}{(\gamma_0 - |b_{n,k}|)^2\sum_{j=1}^k \left(\frac{j}{k+1}\right)^{-\rho_c}}\sim \frac{(1-\rho_c)(\gamma_0 + |b_{n,k}|)^2}{k(\gamma_0 - |b_{n,k}|)^2} \to 0,\qquad k\to\nf,
\]
and, similarly,
\[\max_{1\leq j \leq k}\frac{(v_j)_{2,2}}{\sum_{j=1}^k (v_j)_{2,2}} \leq 
\frac{(\gamma_0 + |b_{n,k}|)^2}{(\gamma_0 - |b_{n,k}|)^2\sum_{j=1}^k \left(\frac{j}{k+1}\right)^{-2\rho_c}}\sim \frac{(1-2\rho_c)(\gamma_0 + |b_{n,k}|)^2}{k(\gamma_0 - |b_{n,k}|)^2} \to 0,\qquad k\to\nf.
\]
Hence, by the Lindberg Feller's CLT,
\[\frac{1}{\sqrt{k}}\sum_{j=1}^k
\begin{pmatrix}
Z_j - \gamma_0 - b_{n,k}\left(\frac{j}{k+1}\right)^{-\rho_0} \\
\left[Z_j - \gamma_0 - b_{n,k}\left(\frac{j}{k+1}\right)^{-\rho_0}\right]\left(\frac{j}{k+1}\right)^{-\rho_c}
\end{pmatrix} \tod \mN_2(0,S_{11}),\qquad k\to\nf,\quad n\to\nf.
\]
and applying the Slutsky's theorem gives
\[\dfrac1{\sqrt{k}}\sum_{j=1}^k \gb_j(Y_j;\thb_0) \tod \mN_2(0,S_{11}).
\]

Recall that from \eqref{el.function} that the empirical log-likelihood function of $\thb$ is given by
\[\ell_E(\thb) = -k\log k - \sum_{j=1}^k \log[1+\lambdab^{\top}\gb_j(Y_j;\thb)],
\]
where $\lambdab$ is the Lagrange multiplier that satisfies 
\eqref{Lagrange} \citep{Owen1990}.
At the MELE $\widehat\thb_E$, we have 
\[\frac{\partial \ell_E(\thb)}{\partial \thb}\bigg|_{\thb=\widehat\thb_E} = \sum_{j=1}^{k} \frac{\frac{\partial \gb_j(Y_j;\widehat\thb_E)}{\partial \thb^{\top}}\widehat\lambdab}{1+\widehat{\lambdab}^{\top}\gb_j(Y_j;\widehat\thb_E)} = \bzero,
\]
where $\widehat\lambdab$ is the Lagrange multiplier at $\widehat\thb_E$ that satisfies
\[\sum_{j=1}^k \frac{\gb_j(Y_j;\widehat\thb_E)}{1+\widehat\lambdab^{\top}\gb_j(Y_j;\widehat\thb_E)} = \bzero.
\]
Define a function of $(\thb,\lambdab)$:
\[Q_k(\thb,\lambdab) = \begin{pmatrix}
Q_{1k}(\thb,\lambdab) \\
Q_{2k}(\thb,\lambdab)
\end{pmatrix} =
\frac{1}{k}\sum_{j=1}^k \begin{pmatrix}
\dfrac{\gb_j(Y_j;\thb)}{1+\lambdab^{\top}\gb_j(Y_j;\thb)} \\
\dfrac{\frac{\partial \gb_j(Y_j;\thb)}{\partial \thb^{\top}}\lambdab}{1+\lambdab^{\top}\gb_j(Y_j;\thb)}\\
\end{pmatrix}.
\]
Let $\bar \gb = (1/k)\sum_{j=1}^k \gb_j(Y_j;\thb_0).$
We then have 
\[Q_k(\thb_0,\zerob) = \begin{pmatrix}
\bar \gb \\ 
\bzero \\
\end{pmatrix}.
\]
By the Taylor expansion of $Q_k$ at $(\thb_0,\zerob)$ (see \cite{QinLawless1994}), we obtain
\[ \bzero = Q_k(\widehat\thb_E,\widehat\lambdab) = Q_k(\thb_0,\zerob) + \begin{pmatrix}
\dfrac{\partial Q_{1k}(\thb_0,\zerob)}{\partial \thb} & \dfrac{\partial Q_{1k}(\thb_0,\zerob)}{\partial \lambdab} \\
\dfrac{\partial Q_{2k}(\thb_0,\zerob)}{\partial \thb} &
\dfrac{\partial Q_{2k}(\thb_0,\zerob)}{\partial \lambdab} \\
\end{pmatrix} \begin{pmatrix}
\widehat\thb_E - \thb_0 \\
\widehat\lambdab - \zerob \\
\end{pmatrix} + o_p(\norm{\widehat\thb_E - \thb_0} + \norm{\widehat\lambdab}).
\]
We next calculate the derivatives of $Q_k$. 
The partial derivatives of $Q_{1k}$ are
\[\frac{\partial Q_{1k}(\thb_0,\zerob)}{\partial \thb} = \frac{1}{k}\sum_{j=1}^k \frac{\partial \gb_j(Y_j;\thb_0)}{\partial \thb} = 
\frac{1}{k}\sum_{j=1}^k
\begin{pmatrix}
-1 & -(\frac{j}{k+1})^{-\rho_c} \\
-(\frac{j}{k+1})^{-\rho_c} & -(\frac{j}{k+1})^{-2\rho_c} \\
\end{pmatrix} = S_{12} + o(1),
\]
where
\[S_{12} = 
\begin{pmatrix}
-1 & -\frac{1}{1-\rho_c} \\
-\frac{1}{1-\rho_c} & -\frac{1}{1-2\rho_c} \\
\end{pmatrix} 
\]
and
\[\frac{\partial Q_{1k}(\thb_0,\zerob)}{\partial \lambdab} = \frac{1}{k}\sum_{j=1}^k 
\gb_j(Y_j;\thb_0)\gb_j^{\top}(Y_j;\thb_0)
= - S_{11} + o(1).
\]
The partial derivatives of $Q_{2k}$ are
\[\frac{\partial Q_{2k}(\thb_0,\zerob)}{\partial \thb}  = \begin{pmatrix}
0 & 0 \\
0 & 0 \\
\end{pmatrix}
\]
and
\[\frac{\partial Q_{2k}(\thb_0,\zerob)}{\partial \lambdab} = \frac{1}{k}\sum_{j=1}^k \frac{\partial \gb_j(Y_j;\thb_0)}{\partial\thb}  = S_{12} + o(1).
\]
To summarize, we have
\[\begin{pmatrix}
\bar \gb \\
\bzero \\
\end{pmatrix} + \begin{pmatrix}
S_{12} & -S_{11} \\
\textbf{0} & S_{12} \\
\end{pmatrix} \begin{pmatrix}
\widehat\thb_E - \thb_0 \\
\widehat\lambdab \\
\end{pmatrix} = o_p(\norm{\widehat\thb_E - \thb_0} + \norm{\widehat\lambdab}).
\]
Since $\sqrt k \bar \gb \tod \mN_2(\bzero, S_{11})$, then $\bar \gb = O_p(k^{-1/2})$, and $\norm{\widehat\thb_E - \thb_0} + \norm{\widehat\lambdab} = O_p(k^{-1/2}).$
Hence,
\begin{equation}
\label{eqn1}
S_{12}(\widehat\thb_E - \thb_0) - S_{11}\widehat\lambdab = -\bar{\gb} + o_p(k^{-1/2}),
\end{equation}
and
\begin{equation}
\label{eqn2}
S_{12}\widehat\lambdab = o_p(k^{-1/2}).
\end{equation}
Multiplying \eqref{eqn1} by $S_{12}S_{11}^{-1}$, we get
\[S_{12}S_{11}^{-1}S_{12}(\widehat\thb_E - \thb_0) = -S_{12}S_{11}^{-1}\bar \gb + o_p(k^{-1/2}).
\]
Therefore,
\[\sqrt{k}(\widehat\thb_E - \thb_0) = -[S_{12}S_{11}^{-1}S_{12}]^{-1}S_{12}S_{11}^{-1}\sqrt{k}\bar{\gb} + o_p(1).
\]
Again using the fact that $\sqrt k \bar \gb  \tod \mN_2(\bzero, S_{11})$, we obtain
\[\sqrt k(\widehat\thb_E - \thb_0) \tod \mN_2(\bzero, \Sigma),
\]
where 
\begin{align*}
\Sigma &= [S_{12}S_{11}^{-1}S_{12}]^{-1}S_{12}S_{11}^{-1}S_{11}S_{11}^{-1}S_{12}[S_{12}S_{11}^{-1}S_{12}]^{-1} 
\\&=
[S_{12}S_{11}^{-1}S_{12}]^{-1}
\\& =
\frac{\gamma_0^2}{\rho_c^2}\begin{pmatrix}
(1-\rho_c)^2 & -(1-\rho_c)(1-2\rho_c) \\
-(1-\rho_c)(1-2\rho_c) & (1-\rho_c)^2(1-2\rho_c) \\
\end{pmatrix}.
\end{align*}
Therefore, as $n\to \infty$,
\[\sqrt{k}(\widehat\gamma_E - \gamma_0) \tod \mN\left(0,\left(\frac{1-\rho_c}{\rho_c}\right)^2\gamma_0^2\right),
\]
and
\[\sqrt{k}(\widehat b_E - b_{n,k}) \tod \mN\left(0,\frac{(1-\rho_c)^2(1-2\rho_c)}{\rho_c^2}\gamma_0^2\right).
\]
\end{proof}

\subsection{Proof of Theorem~\ref{tLR}}\label{sA2}

\begin{proof} 

From \eqref{el.function}, the empirical log-likelihood function of $\thb$ is 
\[\ell_E(\thb) = -k\log k - \sum_{j=1}^k \log[1+\lambdab^{\top}\gb_j(Y_j;\thb)],
\]
where $\lambdab$ is the Lagrange multiplier satisfying
\begin{equation}
\sum_{j=1}^k \frac{\gb_j(Y_j;\thb)}{1+\lambdab^{\top}\gb_j(Y_j;\thb)} = \bzero.
\end{equation}
Since $\dim(\gb_j) = \dim(\thb)$, at $\thb = \widehat\thb_E$, the solutions to the optimization problem in \eqref{optimization} are $p_j = 1/k$, $j=1,\ldots, k$.
Hence $\max_{\thb}\ell_E(\thb) = \ell_E(\widehat\thb_E) = -k\log k$.

Let $(\gamma_0,\tilde{b})$ denote the MELE under $H_0: \gamma = \gamma_0$, and $\tilde{\lambdab}$ the Lagrange multiplier at $\thb = (\gamma_0,\tilde{b})$, which is a $2\times 1$ vector.
Hence, $R(\gamma_0) = 2\sum_{j=1}^k \log[1+\tilde{\lambdab}^{\top}\gb_j(Y_j;\gamma_0,\tilde{b})].$ 

Furthermore, 
\[\frac{\partial \ell_E(\gamma_0,b)}{\partial b}\Bigg |_{b=\tilde b} = 
\sum_{j=1}^k \frac{ \left(\frac{\partial \gb_j(Y_j;\gamma_0,\tilde{b})}{\partial b}\right)^{\top}\tilde{\lambdab}}{1+\tilde{\lambdab}^{\top}\gb_j(Y_j;\gamma_0,\tilde{b})} = 0.
\]
Due to \eqref{Lagrange}, we also have
\[\sum_{j=1}^k \frac{\gb_j(Y_j;\gamma_0,\tilde{b})}{1+\tilde{\lambdab}^{\top}\gb_j(Y_j;\gamma_0,\tilde{b})} = \bzero.
\]
Define a function of $b$ and $\lambdab$:
\[P_k(b,\lambdab) = \dfrac{1}{k}\sum_{j=1}^k \begin{pmatrix}
\dfrac{\gb_j(Y_j;\gamma_0,b)}{1+\lambdab^{\top}\gb_j(Y_j;\gamma_0,b)}\\
\dfrac{\left(\frac{\partial \gb_j(Y_j;\gamma_0,b)}{\partial b}\right)^{\top}\lambdab }{1+\lambdab^{\top}\gb_j(Y_j;\gamma_0,b)}\\
\end{pmatrix}.
\]
Note that 
\[P_k(b_{n,k},\zerob) = \begin{pmatrix}
\bar{\gb} \\ 
0 \\
\end{pmatrix}\quad\text{with } \bar\gb= (1/k)\sum_{j=1}^k \gb_j(Y_j;\g_0, b_{n,k}).
\]
By a similar Taylor expansion of $P_k$ at $(\thb_0,\bzero)$ and calculations as in the proof of Theorem~\ref{asymptotic.normality},  we get 
\[\bzero = P_k(\tilde{b},\tilde{\lambdab}) =  \begin{pmatrix}
\bar{\gb} \\ 
0 \\
\end{pmatrix} + 
\begin{pmatrix}
A_{12} & -S_{11} \\
\bzero & A_{12} \\
\end{pmatrix}
\begin{pmatrix}
\tilde{b} - b_{n,k}\\
\tilde{\lambdab}\\
\end{pmatrix} + o_p(\norm{\tilde{b} - b_{n,k}} + \norm{\tilde{\lambdab}}),
\]
where 
\[A_{12} = \begin{pmatrix}
-\frac{1}{1-\rho_c} \\
-\frac{1}{1-2\rho_c}\\
\end{pmatrix}. 
\]
Since $\bar{\gb} = O_p(k^{-1/2})$, then similar to the proof of Theorem~\ref{asymptotic.normality} we see that
\[\norm{\tilde{b} - b_{n,k}} + \norm{\tilde{\lambdab}} = O_p(k^{-1/2}).
\]
Hence,
\begin{equation}
\label{eqn3}
A_{12}(\tilde{b} - b_{n,k}) - S_{11}\tilde{\lambdab} = -\bar{\gb} + o_p(k^{-1/2})
\end{equation}
and
\begin{equation}
\label{eqn4}
A_{12}^{\top}\tilde{\lambdab} = o_p(k^{-1/2}).
\end{equation}
Define $I_d$ as the identity matrix with rank $d$.
Through similar derivations as in the proof of Theorem~\ref{asymptotic.normality},
\[\tilde{b} - b_{n,k} = -[A_{12}S_{11}^{-1}A_{12}]^{-1}A_{12}S_{11}^{-1}\bar{\gb}.
\] 
Multiplying both sides of \eqref{eqn3} by $A_{12}S_{11}^{-1},$ and substituting this expression for $\tilde{b},$
we get
\[\tilde{\lambdab} = [I_2-A_{12}(A_{12}^{\top}S_{11}^{-1}A_{12})^{-1}A_{12}^{\top}S_{11}^{-1}]\ \bar{\gb} + o_p(k^{-1/2}).
\]
Therefore,
\begin{align*}
R(\gamma_0) &=
2\sum_{j=1}^k \log[1+\tilde{\lambdab}^{\top}\gb_j(Y_j;\gamma_0,\tilde{b})] 
\\&= 2\tilde{\lambdab}^{\top}\sum_{j=1}^k \gb_j(Y_j;\gamma_0,\tilde{b}) - \tilde{\lambdab}^{\top}\sum_{j=1}^k \gb_j(Y_j;\gamma_0,\tilde{b})\gb_j^{\top}(Y_j;\gamma_0,\tilde{b}) \tilde{\lambdab} + o_p(1)
\\&= (\sqrt{k}S_{11}^{-1/2}\bar{\gb})^{\top}\{S_{11}^{-1/2}[I_2-A_{12}(A_{12}^{\top}S_{11}^{-1}A_{12})^{-1}A_{12}^{\top}]S_{11}^{-1/2}\}\sqrt{k}S_{11}^{-1/2}\bar{\gb} + o_p(1).
\end{align*}

Let $H = S_{11}^{-1/2}[A_{12}(A_{12}^{\top}S_{11}^{-1}A_{12})^{-1}A_{12}^{\top}]S_{11}^{-1/2}.$
Here $A_{12}^{\top}S_{11}^{-1}A_{12}$ is a scalar. So, \[{\rm rank}(A_{12}(A_{12}^{\top}S_{11}^{-1}A_{12})^{-1}A_{12}^{\top}) = {\rm rank}(A_{12}A_{12}^{\top}) = 1.
\] 
Since $S_{11}^{-1/2}$ has full rank, then ${\rm rank}(H) = {\rm rank}(A_{12}A_{12}^{\top}) = 1.$
We have $H^{\top} = H$ and 
\begin{align*}
H^2 &=  S_{11}^{-1/2}A_{12}(A_{12}^{\top}S_{11}^{-1}A_{12})^{-1}(A_{12}^{\top}S_{11}^{-1}A_{12})(A_{12}^{\top}S_{11}^{-1}A_{12})^{-1}A_{12}^{\top}S_{11}^{-1/2}
\\&= S_{11}^{-1/2}[A_{12}(A_{12}^{\top}S_{11}^{-1}A_{12})^{-1}A_{12}^{\top}]S_{11}^{-1/2}
\\&= H.
\end{align*}
Hence, $H$ is symmetric and idempotent. Therefore,
\[I_2 - H = I_2-S_{11}^{-1/2}[A_{12}(A_{12}^{\top}S_{11}^{-1}A_{12})^{-1}A_{12}^{\top}]S_{11}^{-1/2}.
\] 
Since $\sqrt{k}S_{11}^{-1/2}\bar{\gb} \tod \mN_2(\bzero,I_2)$, then by the Cochran's theorem \citep{Cochran1934}, $R(\gamma_0) \tod \chi_1^2.$
\end{proof}

\end{document}